\documentclass[11pt]{article}
\usepackage{chngcntr}
\usepackage[papersize={8.5in,11in},top=1.1in,left=1.1in,right=1.1in,bottom=1.1in]{geometry}

\usepackage[final]{graphicx}
\usepackage{amssymb}
\usepackage{amsmath,amsthm}
\usepackage[boxed]{algorithm2e}
\usepackage{appendix}
\usepackage{bm}
\usepackage{hyperref}
\usepackage{multirow}
\usepackage{cite}
\usepackage{subfigure}

\newenvironment{definition}[1][Definition]{\begin{trivlist}
\item[\hskip \labelsep {\bfseries #1}]}{\end{trivlist}}

\newtheorem{theorem}{Theorem}[section]

\newtheorem{exampl}[theorem]{Example}
\newtheorem{lemma}[theorem]{Lemma}
\newtheorem{claim}[theorem]{Claim}
\newtheorem{proposition}[theorem]{Proposition}
\newtheorem{corollary}[theorem]{Corollary}

\newtheorem*{inf_thm}{(Informal Theorem)}{\bfseries}{\itshape}

\title{Pricing to Maximize Revenue and Welfare Simultaneously in Large Markets}
\author{Elliot Anshelevich \and Koushik Kar \and Shreyas Sekar}

\begin{document}


\maketitle

\begin{abstract}
We study large markets with a single seller which can produce many types of goods, and many multi-minded buyers. The seller chooses posted prices for its many items, and the buyers purchase bundles to maximize their utility. For this setting, we consider the following questions: What fraction of the optimum social welfare does a revenue maximizing solution achieve? Are there pricing mechanisms which achieve both good revenue and good welfare simultaneously? To address these questions, we give efficient pricing schemes which are guaranteed to result in both good revenue and welfare, as long as the buyer valuations for the goods they desire have a nice (although reasonable) structure, e.g., that the aggregate buyer demand has a monotone hazard rate or is not too convex. We also show that our pricing schemes have implications for any pricing which achieves high revenue: specifically that even if the seller cares only about revenue, they can still ensure that their prices result in good social welfare without sacrificing profit. Our results holds for general multi-minded buyers in large markets; we also provide improved guarantees for the important special case of unit-demand buyers.
\end{abstract}

\setcounter{page}{1}

\section{Introduction}

Social Welfare and Profit~\footnote{For convenience, we will use revenue and profit interchangeably in this work} are the two canonical objectives in the extensive literature dealing with \emph{envy-free} algorithmic pricing. The study of these two objectives, in isolation from each other, has inspired the design of novel pricing mechanisms for revenue maximization~\cite{balcanBM08,guruswami2005profit} in a variety of interesting markets, and an equally impressive body of work on welfare maximization~\cite{chen2014revenue,deng2013pricing,hsuMRRV16}. While the significance of profit and social welfare is undisputed, it is easy to overlook the fact the two objectives do not exist in a vacuum. For instance, although a monopolistic seller may only be interested in profits, myopically increasing prices while compromising on buyer welfare can lead to poor long-term revenue. This is distinctly true for large markets with repeated engagement where singularly optimizing for one objective while ignoring the other (as in the existing literature) could adversely affect the health of the marketplace~\cite{bachrachCKKK14}. Therefore, not only is it desirable to promote the design of holistic pricing solutions that optimize on both counts simultaneously, it is also crucial to gain a better understanding of how existing algorithms perform in a \textit{bicriteria} sense. Against this backdrop, we seek to address the following questions.
\begin{quote}

{\em What fraction of the optimum social welfare does a revenue maximizing solution achieve? Are there pricing mechanisms which achieve both good revenue and good welfare simultaneously? }

\end{quote}

Both in economics and in computer science~\cite{kleinbergY13}, it is well understood that the goals of maximizing revenue and social welfare are often at odds with each other. Bearing this in mind, we seek to quantify the exact amount of friction between these two objectives in large markets. In particular, we are interested in understanding the surplus achieved by a profit maximizing solution, a problem that has received considerable attention in Auction theory~\cite{abhishekH10, dughmiRS12,kleinbergY13}. The fact that we restrict our attention to the revenue end of the spectrum is motivated partly by the observation that welfare maximizing prices can result in negligible profits (see Example~\ref{ex:socmaxpoorrev}) even for trivial instances. However, unlike most analogous work in the theory of auctions, we are interested in understanding these trade-offs as well as designing bicriteria approximation algorithms in multi-item markets where the seller's modus operandi involves posting prices on the individual goods. In that sense, this work is a high-level extension of the recent body of work on envy-free revenue-maximization~\cite{anshelevichKS15,briest2008uniform,guruswami2005profit} towards additional ambitious objectives.

\subsection{Market Model: Item Pricing for Multi-Minded Buyers}

In this work, we adopt a simple posted-pricing mechanism that captures the operation of most real-life large markets: the seller posts a single price per good, and each buyer purchases a bundle of goods that maximizes their utility. The seller controls a set $T$ of available goods, and can produce any desired quantity $x_t$ of a good $t \in T$, for which he incurs a cost of $C_t(x_t)$. The market consists of many, many buyers who are \emph{multi-minded}, meaning that each buyer $i$ has a `desired set of bundles of goods': the buyer has the same value $v_i$ for each of these bundles and under a given set of prices, purchases the bundle that maximizes her utility.

The market model that we study in this work is reasonably general. Multi-minded buyers represent a class of computationally attractive yet combinatorially non-trivial buyer valuations that have recently been featured in a number of papers~\cite{buchbinderG15,krystaTV15}. Perhaps, more importantly, the class strictly generalizes highly popular models such as \emph{unit-demand} and \emph{single-minded} valuations. Secondly, the convex production costs considered in our framework strictly generalize models with limited (or unlimited) supply, which are usually the norm in the pricing literature. As~\cite{anshelevichKS15, blumGMS11} point out, limited supply is often too rigid for realistic, large markets where the seller may be able to increase production, albeit at a higher cost. Bicriteria algorithms notwithstanding, our work actually presents the first profit-maximization algorithms for general multi-minded buyers even with limited supply.

Our model captures several scenarios of interest wherein a typically profit maximizing seller may be driven to ensure good overall social utility. We illustrate two of them here.
\begin{enumerate}
\item \textbf{PEV Charging}: In a market capturing charging stations for plug-in electric vehicles, each good represents a time slot, and each buyer may desire specific (sets of) time slots based on time constraints and charging capacity. Varying demand and electricity generation costs necessitate differential pricing across time slots~\cite{bhattacharya2014extended,tushar2012economics}.

\item \textbf{Advertising Markets}: A publisher in such a market may decide to use simultaneous posted prices to auction off a set of distinct ad-items (positions or locations on the website) to buyers interested in purchasing specific subsets of these items to reach target audiences.
\end{enumerate}

\subsubsection*{Circumventing Computational Complexity via Oblivious Guarantees}

One of the challenges in essentially any non-trivial setting (including \emph{all} the settings which we consider), is that computing profit-maximizing prices is NP-Hard. This is largely due to the fact that the seller is not allowed to price-discriminate, i.e., it must charge the same price for each good to all the buyers, instead of having different prices for each buyer. In view of the computational barriers surrounding the optimal profit solution, a seller which cares about profit may use a variety of strategies, from approximation algorithms to heuristics. The uncertainty regarding the actual strategy adopted by the seller in turn casts aspersions on the practical significance of our goal of characterizing the social welfare at optimal-revenue solutions. One of the contributions of this work is a simple but powerful framework that allows us to completely circumvent the complexity question: \emph{our guarantees on the social welfare do not depend on the exact details of the pricing mechanism used by the seller, and instead would hold for a wide variety of pricing mechanisms, as long as these prices achieve decent revenue guarantees.}

\subsubsection*{Inverse Demand Functions and $\alpha$-Strong Regularity}
In order to concisely represent the large number of buyers in the market, we classify the buyers into a finite set of buyer types $B$ such that all of the multi-minded buyers belonging to a certain type desire the same set of bundles. Then, each buyer type can be fully captured by a subset of $2^T$ along with an \emph{inverse demand distribution} $\lambda_i(x)$ describing the valuations of buyers having this type. Formally, for any buyer type $i \in B$, $\lambda_i(x) = p$ implies exactly $x$ amount of buyers of type $i$ have a valuation of $p$ or more for each bundle in their common desired set. Although different buyer types may have different demand functions, it is natural to assume that the valuations of all buyers are often sampled (albeit differently) from some global distribution. Because of this, we will make the assumption that the buyer valuations for every type have the same support $[\lambda^{min}, \lambda^{max}]$. 

A first stab at the problem reveals that the above framework is too coarse to obtain meaningful trade-offs between social welfare and profit. Indeed, it is not hard to reason that a precise characterization of the revenue-welfare trade-offs would depend heavily on the distributions of the buyer valuations. To better understand this dependence, we study a class of inverse demand functions parameterized by a single parameter $\alpha \in [0,1]$ known as $\alpha$-strongly regular distributions.

\begin{definition}(\textbf{$\alpha$-Strong Regularity}\cite{coleR14})
\label{def:strongregularity}
A buyer type $i$ is said to have an $\alpha$-Strongly regular demand function ($\alpha$-SR) for $\alpha \in [0,1]$ if for any $x_1 < x_2$, we have $\frac{\lambda_i(x_2)}{|\lambda'_i(x_2)|} -  \frac{\lambda_i(x_1)}{|\lambda'_i(x_1)|} \leq \alpha(x_2 - x_1)$.
\end{definition}

\noindent $\alpha$-Strongly regular distributions were introduced in~\cite{coleR14} as a strict generalization of \emph{monotone hazard rate} (MHR) distributions that smoothly interpolate between MHR $(\alpha = 0$) and regular distributions $(\alpha = 1)$. Our main contribution is the design of mechanisms that simultaneously obtain good revenue and welfare for small $\alpha$, and degrade gracefully as $\alpha$ increases. Note that even the set of $\alpha$-SR functions with $\alpha=0$, for which we obtain the strongest results, contains a large class of important distributions, including exponential (e.g., $e^{-x}$), polynomial (e.g., $1-x^2$), and all log-concave functions. The reader is asked to refer to Section~\ref{sec:prelim} for a more detailed discussion regarding this class of distributions.

\subsection{Our Contributions}
The primary algorithmic contribution of this work is a new $($profit, welfare$)$-bicriteria approximation for general markets with multi-minded buyers and production costs, stated below.

\begin{inf_thm}
We can compute in poly-time a set of item prices that guarantee a $(\Theta(\frac{ \log \Delta}{1-\alpha}))$-approximation to the optimum profit, and a $\Theta(\frac{1}{1-\alpha})$-approximation for welfare, where $\Delta$ is the ratio of the size of the largest bundle desired by any buyer to the smallest one.
\end{inf_thm}

\noindent There are several exciting aspects to this result: $(i)$ Not only do we present the \emph{only-known bi-criteria approximation} algorithm for such general markets, but ours is also the first non-trivial profit-maximization algorithm for multi-minded buyers in the envy-free literature.
$(ii)$ Our welfare guarantees are completely independent of the bundle size ($\Delta$). $(iii)$ When buyers desire small bundles (e.g., for unit-demand valuations where all bundles are unit-size), our solution extracts a constant portion of the social welfare as revenue, as illustrated in Figure~\ref{figure:bicriteria_unitdemand}. Moreover, for the important special case of unit-demand valuations, we provide a much simpler pricing mechanism with slightly better constant approximation factors than in the theorem statement. $(iv)$ Finally, even when buyers desire large bundles, it is reasonable to expect that in markets with similar types of goods, the various bundles are of approximately the same size, i.e., $\Delta$ is small. In the PEV example, one expects different electric vehicles to have similar charging capacity. This case with similar-sized bundles is a considerably non-trivial instance for which the revenue-welfare gap becomes small.

%
%

%
%

\begin{figure}
\centering
\hfill
\subfigure[Exact bounds for profit and welfare as a function of $\alpha$ for unit-demand buyers. We obtain constant-factor bicriteria approximations when the demand is close to \emph{MHR} ($\alpha = 0$), and good guarantees for larger $\alpha$ even when the demand is close to the \emph{equal-revenue distribution} ($\alpha=1$).]{\label{figure:bicriteria_unitdemand} \includegraphics[width = 0.4\linewidth]{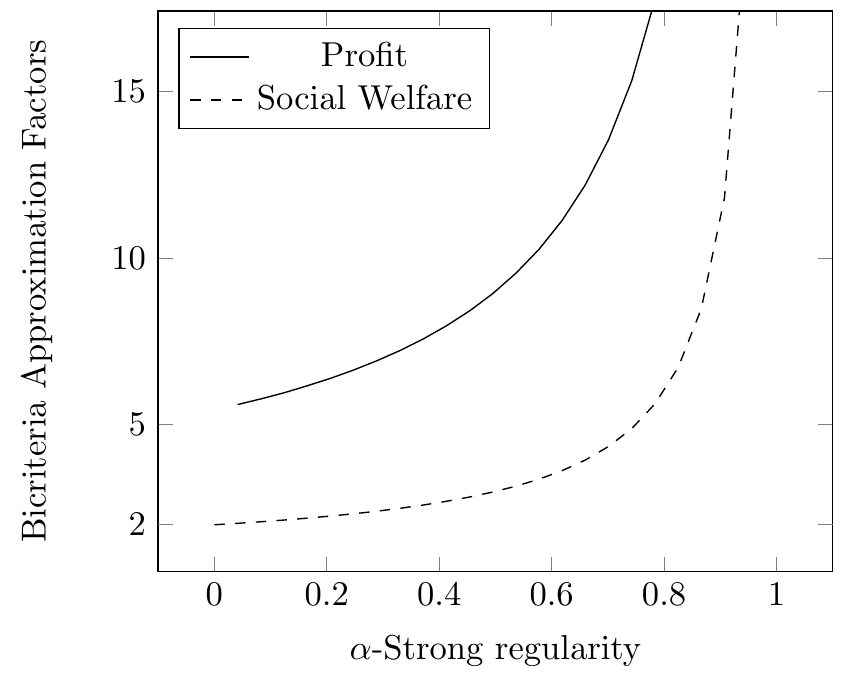}}
\hfill
\subfigure[Actual Revenue-Welfare Curve for unit-demand buyers, $\alpha = 0$:  The exact bicriteria guarantees lie on the trade-off curve, so that one of the two approximation factors is significantly better than the worst-case bound (point $X$). For example, when welfare is only half-optimal, the revenue factor improves from $2e$ to $2$.]{\label{figure:tradeoff} \includegraphics[width=0.4\linewidth]{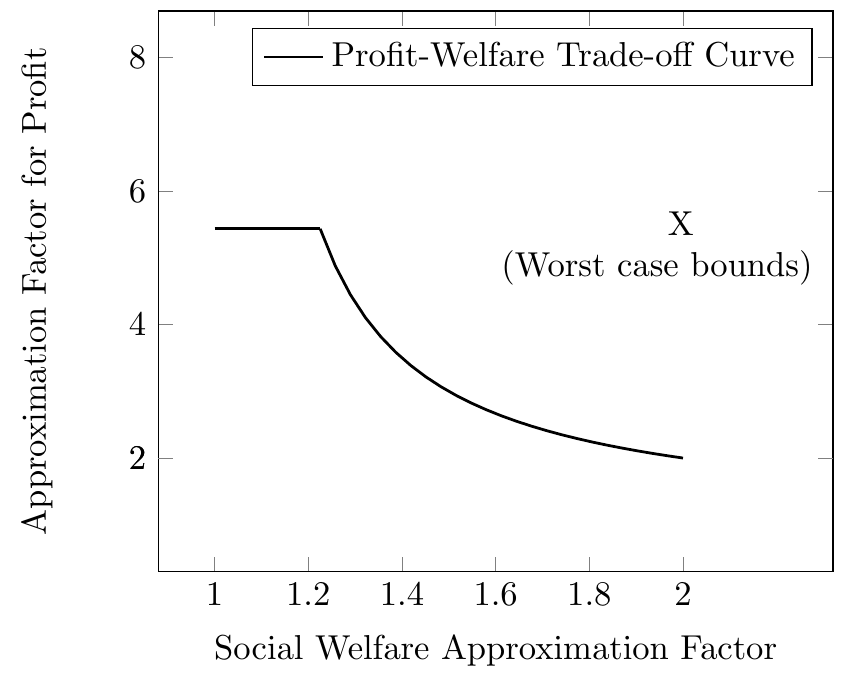}}
\hfill
\end{figure}


\textbf{Profit-Welfare Trade-Offs:} The approximation guarantees in the theorem statement are only the worst-case bounds derived independently for each objective. In fact, as illustrated in Figure~\ref{figure:tradeoff}, we prove that the two worst-case factors never occur simultaneously and the actual bicriteria bound lies on a trade-off curve, resulting in improved approximations for at least one objective, i.e., if the actual welfare is close to the worst-case guarantee, then the profit is much better than in the theorem and vice-versa.
%
%
	
\subsubsection*{Social Welfare of other Revenue-Maximizing Solutions}
All of the revenue guarantees in this paper are shown by comparing the profit of our solution to the optimum welfare, an approach that has strong implications towards bounding the social welfare of other profit-maximizing solutions. Specifically, we design a simple framework using our bicriteria bounds as a black-box result and show that {\em any} pricing mechanism which achieves revenue better than our (efficiently computable) mechanism is guaranteed to deliver at least a $\Theta(\frac{\log \Delta}{1-\alpha})$-approximation to the optimum welfare. This holds whether the seller computes revenue-maximizing prices (an NP-Hard problem even for unit-demand with $\alpha=0$), or uses a more efficient mechanism. Thus, one of the main messages of this paper is that even a seller interested solely in maximizing profits can guarantee a good social welfare without sacrificing any revenue. For example, in unit-demand markets with MHR valuations, there is absolutely no excuse for such a seller to not also achieve at least a $2e$-approximation to the optimum social welfare irrespective of their preferred pricing mechanism.

\noindent\textbf{Technical Difficulties:}
Although our large market model falls in the realm of settings where it is possible to efficiently compute social welfare maximizing prices, exploiting this for profit-maximization as in~\cite{balcanBM08,guruswami2005profit, cheung2008approximation} leads to poor approximation guarantees, for e.g., $O(\frac{\lambda^{max}}{\lambda^{min}})$-bounds even for unit-demand instances with no production costs. Instead, our techniques rely crucially on exploiting the structure of $\alpha$-strong regular functions to efficiently compute prices that compromise neither on revenue nor welfare; following this, we also develop new computational insights for characterizing profit via ascending-price procedures in settings with multi-minded buyers and production costs. Finally, in this work, we will focus solely on deterministic pricing mechanisms. While randomized mechanisms that mix between welfare and profit maximizing solutions are a theoretically fascinating apparatus, the ensuing price fluctuations and alternating buyer dissatisfaction render them unsuitable for many settings of interest~\cite{diakonikolasPPS12,likhodedovS04}

\subsubsection*{Related Work: Existing Bicriteria Algorithms}
The primary barrier towards designing envy-free revenue-maximizing prices --- a lack of insight regarding the optimum solution --- is also, surprisingly, the chief architect behind the existence of many (implicit) bi-criteria approximation algorithms in the algorithmic pricing literature. More concretely, a majority of the revenue-maximization algorithms in the literature~\cite{balcanBM08, briest2011buying, chawla2007algorithmic, cheung2008approximation, guruswami2005profit, ImLW12} achieve their approximation factors for revenue by comparing it to the optimum social welfare. Exploiting these revenue-welfare ties further, it is not hard to see 
that such an $\alpha$-approximation algorithm for revenue trivially results in a $(\alpha, \alpha)$-bicriteria approximation. For instance, the results from~\cite{guruswami2005profit} and~\cite{balcanBM08} immediately imply $(\Theta(\log |B|), \Theta(\log |B|))$-bicriteria approximation algorithms for unit-demand and unlimited supply markets respectively.

Our results improve upon the work mentioned above across multiple dimensions. In contrast to the trivial $(\alpha,\alpha)$ type bounds in previous work, our specific focus on bicriteria approximations leads to significant improvements in social welfare without sacrificing much profit. We also remark that while specific bicriteria bounds were also provided in~\cite{anshelevichKS15}, their results only apply to the easiest version of our setting ($\alpha = 0, \Delta =1$). Indeed, we study general settings with multi-minded buyers and production costs, instead of unit-demand and single-mined valuations with limited supply, which are much more common in algorithmic pricing literature. Finally, our bounds are more nuanced due to their dependence on $\alpha$, and independence from $|B|$, e.g., $\Theta(\log |B|)$-bounds are totally unacceptable for large markets. 

\subsection{Other Related Work}
While $($revenue,welfare$)$-bicriteria approximation algorithms have not specifically been tackled in the pricing literature beyond the single good case, the broader understanding of trade-offs between the two objectives has been a prominent motif that has cropped up time and again in various forms. We first highlight a few overarching differences between our results and other work on  profit-welfare trade-offs, specifically in auctions: $(i)$ we study reasonably general combinatorial markets with multi-minded buyers and production cost functions, and not just limited-supply settings with unit-demand buyers as in other work,  $(ii)$ our results do not depend on externalities such as tuning the objective function or the addition of external buyers as in Bulow-Klemperer type theorems~\cite{AggarwalGM09,bulow1994auctions,roughgarden2007efficiency}, and finally $(iii)$ unlike similar (types of) results in Bayesian auctions, our pricing mechanisms are non-discriminatory, and therefore, envy-free.

Characterizing the efficiency of revenue-optimal mechanisms is an extremely fundamental question that has spurred multiple avenues of research starting with Bulow and Klemperer's seminal result~\cite{bulow1994auctions} that adding a single buyer is usually more profitable than blindly optimizing revenue. Most pertinent to the questions posed in this work are the tight bounds on the (in)efficiency of the Myerson revenue-maximizing mechanism for single good settings appearing in~\cite{abhishekH10, kleinbergY13}: in particular,~\cite{kleinbergY13} provides welfare bounds for general single-parameter auctions as a function of the distribution of buyer valuations, as we do in this work.

Moving beyond welfare lower bounds, other researchers have adopted a more constructive approach by designing explicitly taking into account both the objectives either via bicriteria mechanisms~\cite{sivanST12, daskalakisP11}, or by optimizing linear combinations of revenue and welfare~\cite{likhodedovS04}, or even characterizing the revenue-welfare Pareto curve~\cite{diakonikolasPPS12}. We reiterate that all of the above papers consider simple single good settings, which are easier to characterize, and where the revenue-optimal mechanism is well understood. Moreover, in comparison to the revenue-welfare Pareto curves in~\cite{diakonikolasPPS12}, the implicit trade-off curves in our work are of a different nature as they are obtained for a single instance on top of the worst-case bounds. Finally, multi-objective trade-offs are quite popular in the Sponsored Search literature~\cite{bachrachCKKK14,deng2013pricing,lucier2012revenue,roughgarden2007efficiency}, where the repeated engagement and the tight competition (between sellers)  necessitates approximately revenue-optimal mechanisms that do not compromise on overall social welfare. Such settings can essentially be viewed as a special case of unit-demand markets.
%
%
%
%
%
%
%
%


\section{Model and Preliminaries}
\label{sec:prelim}
Our market model comprises of a single seller controlling a set $T$ of goods and a large number of infinitesimal buyers. The buyers can be concisely represented using a finite set of multi-minded buyer types $B$: for a given type $i \in B$, all the buyers having this type desire the same set of item bundles $B_i \subseteq 2^T$, and each buyer is indifferent between the bundles in $B_i$. Notice that when all of the desired bundles are of unit cardinality, our model reduces to the \textit{unit-demand} case; when each buyer type desires only a single bundle $(|B_i| = 1)$, we get \textit{single-minded} valuations. For general multi-minded valuations, we also assume the free disposal property. Finally, buyers belonging to the same type may hold different valuations for the same bundles, this is modeled by way of an \textit{inverse demand function} $\lambda_i(x)$ for every $i \in B$; $\lambda_i(x) = p$ implies that exactly $x$ amount of buyers of type $i$ value the bundles in $B_i$ at valuation $p$ or more. Given $\lambda_i(x) = p$, it is not hard to see that the total utility derived when $x$ amount of buyers purchase some bundle at price $p$ is $u_i(x) = \int_{z=0}^x \lambda_i(z)dz.$

The market operates according to a natural pricing mechanism with the seller posting a price $p_t$ for each good $t \in T$. Buyers purchase one of the utility-maximizing bundles available to them, i.e., a buyer belonging to type $i$ will purchase the cheapest bundle in $B_i$ as long as its price is no larger than her valuation for the same. Therefore, if $\bar{p}_i$ denotes the bundle in $B_i$ with the smallest price and $\bar{x}_i$ is the population of buyers of this type who purchased some bundle, then $\lambda_i(\bar{x}_i) = \bar{p}_i$.

\subsubsection*{Pricing Solutions, Social Welfare, and Revenue}

We use $(\vec{p}, \vec{x}, \vec{y})$ to represent the outcome of the market mechanism. Here $\vec{p}$ is the vector of prices, $\vec{x}$ denotes the allocation to the buyers with $x_i$ being the total amount of good purchased by buyers of type $i$, and finally $y_t$ is the total amount of good $t \in T$ sold to the buyers. We now define the two main metrics that form the crux of this paper.

\begin{definition}
The social welfare of a solution $(\vec{p}, \vec{x}, \vec{y})$ is defined to be the total utility of all the buyers and the seller and therefore, is equal to the utility of the buyers minus the production cost incurred by the seller, i.e.,

$$SW(\vec{p}, \vec{x}, \vec{y}) = \sum_{i \in B}\int_{x=0}^{x_i} \lambda_i(x)dx - \sum_{t \in T}C_t(y_t).$$
\end{definition}

Notice that the social welfare is independent of the prices, and depends only on $(\vec{x}, \vec{y})$.

\begin{definition}
The (seller's) profit at the solution $(\vec{p}, \vec{x}, \vec{y})$ is the total income due to each good in the market minus the total production cost incurred, i.e.,

$$\pi(\vec{p}, \vec{x}, \vec{y}) = \sum_{t \in T}[p_ty_t - C_t(y_t)].$$
\end{definition}

One of the main goals of this paper is to obtain a lower bound on the social welfare of the profit maximizing solution. We will use $(\vec{p^{opt}}, \vec{x^{opt}}, \vec{y^{opt}})$ to denote the maximum profit solution, and $(\vec{p^{*}}, \vec{x^{*}}, \vec{y^{*}})$ to denote the solution maximizing welfare. Sometimes we will also use $SW^*=SW(\vec{p^{*}}, \vec{x^{*}}, \vec{y^{*}})$ and $\pi^{opt}=\pi(\vec{p^{opt}}, \vec{x^{opt}}, \vec{y^{opt}})$. Thus the quantity we are interested in is $\frac{SW(\vec{p^*},\vec{x^*}, \vec{y^*})}{SW(\vec{p^{opt}}, \vec{x^{opt}},\vec{y^{opt}})}$.

We remark that the maximum social welfare solution should ideally be represented as $(\vec{x^*}, \vec{y^*})$. However, due to the following simple claim (inspired by a similar claim from~\cite{anshelevichKS15}), we will define $p^*_t$ as below, and let $(\vec{p^{*}}, \vec{x^{*}}, \vec{y^{*}})$ denote this specific welfare-maximizing solution. More importantly, as we show in the Appendix, such a welfare maximizing solution can be computed efficiently via a simple convex program.

\begin{claim}
There exists a welfare-maximizing solution $(\vec{p^*},\vec{x^*},\vec{y^*})$ where for every good $t \in T$, $p^*_t = c_t(y^*_t)$. (Here $c_t$ is the derivative of $C_t$.)
\end{claim}

%
%
\subsubsection*{Structure of the Demand and Cost Functions} In this work, we are interested in studying markets with many, many buyers and therefore, correspondingly we take both the inverse demand and production cost functions to be continuously differentiable. In addition, we also make the standard assumption that the utility function $u_i(x)$ is concave for all $i \in B$ and therefore its derivative $\lambda_i(x)$ is non-increasing with $x$. Finally, we consider production cost functions that are \emph{doubly convex} i.e, both $C_t$ and its derivative $c_t(x) = \frac{d}{dx}C_t(x)$ are convex and non-decreasing for all $t \in T$ and further, $C_t(0) = c_t(0) = 0$. A number of well studied cost functions fall within our framework~\cite{blumGMS11}.

As mentioned previously, it is often natural to assume that the inverse demand distributions for different buyer types have the same support $[\lambda^{min}, \lambda^{max}]$. In fact, \emph{all} of our results hold under the more general \emph{uniform peak} assumption, which will be assumed for the rest of this work.

\begin{definition}{(Uniform Peak Assumption)}
For every $i \in B$, $\lambda_i(0) = \lambda^{max}$.
\end{definition}

\noindent\textbf{$\alpha$-SR inverse demand} (Definition~\ref{def:strongregularity}) Recently, $\alpha$-strong regularity has gained some popularity~\cite{coleR14,coleR15,huang2015making} as an elegant characterization of the class of regular functions, which encompasses most well-studied demand distributions including polynomial ($\lambda(x) = 1-x^2$; $\alpha=0$), exponential ($\lambda(x) = e^{-x}$; $\alpha=0$), power law ($\lambda(x) = \frac{1}{\sqrt{x}}$; $\alpha=\frac{1}{2}$), and the equal-revenue distribution ($\lambda(x) = \frac{1}{x}$; $\alpha=1$). For such functions, one can interpret $\alpha$ as a measure of the convexity of the function as larger values as $\alpha$ imply greater convexity or alternatively as a bound on the volatility of the inverse demand $\lambda_i(x)$ as every $\alpha$-SR demand function satisfies  $\frac{d}{dx}\left(\frac{\lambda(x)}{|\lambda'(x)|}\right) \leq \alpha$. As expected, the equal-revenue distribution $(\alpha=1)$ leads to the worst-case bounds for all of our results: in fact even in single good, single buyer type markets, the revenue-optimal solution for $\alpha=1$ only extracts a negligible fraction of the optimum social welfare. However, what is surprising (as evidenced by Figure~\ref{figure:bicriteria_unitdemand}) is that we obtain reasonably good performance guarantees even when $\alpha$ is larger than $\frac{1}{2}$. Note that even $\alpha$-SR functions with very small $\alpha$ include a very large class of interesting demand functions (see above), including any log-concave function.

\section{Warm-up: Profit and Welfare for Unit-Demand Markets}
\label{sec:unitdemandresult}
As a first step towards stating our general results in Section \ref{sec:multi-minded}, we consider the important special case of unit-demand markets, perhaps the most popular class of valuations in the algorithmic pricing literature~\cite{briest2011buying, chawla2007algorithmic, guruswami2005profit}.
Recall that unit-demand valuations are a simple sub-class of multi-minded functions, wherein for each buyer (type) $i \in B$, the bundles desired by $i$ are singleton sets. Therefore, we will overload notation, and when talking about unit-demand buyers say that $B_i \subseteq B$ is the set of items acceptable by all of the buyers of type $i$. Our main result in this section is a simple pricing rule that achieves a $(\Theta(\frac{1}{1-\alpha}), \frac{2-\alpha}{1-\alpha})$-bicriteria approximation algorithm for revenue and welfare respectively when the buyers have $\alpha$-SR inverse demand functions. While we present a generalized version of this theorem in Section~\ref{sec:multi-minded}, the simple techniques from this section do not really extend to multi-minded case. In fact, we require several new gadgets and a more sophisticated pricing mechanism to achieve the generalization.

That said, our reasons for dedicating an entire section to unit-demand buyers is two-fold: $(i)$ the algorithm presented in this section is extremely simple and the constant factors hidden by the asymptotic bound are smaller, and $(ii)$ the unit-demand case provides a platform for us to build intuition and more importantly, discuss the various implications of our results including the profit-welfare trade-off and the ability to derive welfare bounds for the profit-maximizing solution. Before stating our main theorem, we give a simple example to highlight the poor revenue guarantees obtained by the welfare maximizing prices even for single-good markets.

%

\begin{exampl}
\label{ex:socmaxpoorrev}
Consider a single good market with a negligible production cost function, say $C(x) = \epsilon x$. Obviously, there is only one buyer type, and suppose that its inverse demand is $\lambda(x) = 1 - x$ ($\alpha$-SR for $\alpha = 0$). It is easy to observe that at the social wefare maximizing solution, the good is priced at $p^* = \epsilon$, resulting in near-zero revenue. On the contrary, one can price at $p^{opt} = \frac{1}{2}$ to obtain a constant fraction of the optimum welfare as profit.
\end{exampl}



\begin{theorem}
\label{thm_mainrevenuewelfare}
For any unit-demand instance where buyers have $\alpha$-strongly regular inverse demand functions, there is a poly-time $(\zeta, \frac{2-\alpha}{1-\alpha})$-bicriteria approximation algorithm for revenue and social welfare respectively, where

$$\zeta = 2(\frac{1}{1-\alpha})^{\frac{1}{\alpha}} + \frac{\alpha}{1-\alpha} = \Theta\left(\frac{1}{1-\alpha}\right).$$
\end{theorem}

The exact guarantees for profit and welfare are illustrated in Figure~\ref{figure:bicriteria_unitdemand}.

\noindent (Algorithm) The bicriteria approximation factor is achieved by the following simple pricing mechanism
\begin{itemize}
\item Compute the max-welfare solution $(\vec{p^{*}}, \vec{x^{*}}, \vec{y^{*}})$.
\item For every good $t$, set its price $\tilde{p}_t = \max{(p^*_t, \lambda^{max} (1-\alpha)^{\frac{1}{\alpha}})}.$
\end{itemize}

\noindent{\bf Proof Sketch:} We prove this theorem in the Appendix, but here we give some intuition for why this produces a good approximation to both profit and welfare. Let $\tilde{p}=\lambda^{max} (1-\alpha)^{\frac{1}{\alpha}}$. We begin by analyzing these types of {\em thresholded} pricing schemes, in which the price is simply the maximum of the optimum price $p^*_t$ and a constant. In Lemma \ref{lem_proxysimilartoopt}, we show that such solutions have nice structure; essentially we can think of buyers who purchase goods priced at $\tilde{p}$ and those who purchase goods priced at $p^*_t$ as separate systems.

\begin{lemma}
\label{lem_proxysimilartoopt}
Suppose that $((\tilde{p})_{t \in T}, (\tilde{x})_{i \in B}, (\tilde{y})_{t \in T})$ is a pricing solution resulting from a thresholded pricing vector. Then,
\begin{enumerate}
\item The market can be clustered into two mutually disjoint sets of buyers and goods $(B^H, T^H)$ and $(B^L, T^L)$ so that the buyers in each cluster only purchase the goods in the same cluster and $(a)$ for $(i,t) \in (B^H,T^H)$, $\tilde{p}_t = p^*_t$, $\tilde{x}_i = x^*_i$, and $\tilde{y}_t = y^*_t$; $(b)$ for $(i,t) \in (B^L,T^L)$, $\tilde{p}_t \geq p^*_t$, $\tilde{x}_i \leq x^*_i$, and $c_t(\tilde{y}_t) \leq c_t(y^*_t)$.

\item $\vec{\tilde{y}}$ is a welfare-maximizing allocation vector with respect to the demand vector $\vec{\tilde{x}}$.
\end{enumerate}
\end{lemma}

We then prove that due to our choice of $\tilde{p}$ the following two bounds hold: $$SW(\vec{\tilde{p}}, \vec{\tilde{x}}, \vec{\tilde{y}})\leq (2(\frac{1}{1-\alpha})^{\frac{1}{\alpha}}-1)\pi(\vec{\tilde{p}}, \vec{\tilde{x}}, \vec{\tilde{y}}),$$ and $$SW(\vec{p^{*}}, \vec{x^{*}}, \vec{y^{*}}) - SW(\vec{\tilde{p}}, \vec{\tilde{x}}, \vec{\tilde{y}}) \leq \frac{1}{1-\alpha}\pi(\vec{\tilde{p}}, \vec{\tilde{x}}, \vec{\tilde{y}}).$$

Once we show those lower bounds on the profit of our pricing scheme, we apply the following simple claim to finish the proof of the theorem. While the difficult parts of our argument lie in proving the above bounds, we state the claim here since, despite its technical simplicity, we believe that this claim and our general approach may find application in other settings.



\begin{claim}
\label{lem_bicritimprov}
Suppose that for some solution $s$ we have that $SW(s) \leq c_1 \pi(s)$, and $SW^* - SW(s) \leq c_2 \pi(s)$. Then,
$s$ is a $(c_1 + c_2, c_2 + 1)$-bicriteria approximation to (profit,welfare); moreover $SW^*\leq (c_1+c_2) \pi(s)$.  $\blacksquare$
\end{claim}

Henceforth, we will unambiguously use $\zeta$ to denote the exact approximation factor appearing in the statement of the above theorem. Interestingly, in the proof of Theorem~\ref{thm_mainrevenuewelfare}, the revenue guarantee of $\zeta$ is shown with respect to the optimum social welfare, which in turn has important consequences for other profit-maximizing solutions (see Section~\ref{sec:implications}). For now, we just state this revenue guarantee formally.

\begin{corollary}
\label{corr_revenuewrtwelfare}
The profit obtained by the pricing mechanism in Theorem \ref{thm_mainrevenuewelfare} is always within a factor $\zeta$ of the optimum social welfare.
\end{corollary}

\subsubsection{Profit-Welfare Trade-offs}
We further exploit the close ties between revenue and social welfare, and present a revenue-welfare trade-off that improves upon the bicriteria bound in Theorem~\ref{thm_mainrevenuewelfare} by showing that at least one of revenue or welfare is better than the factor guaranteed by the theorem. The bounds in Theorem \ref{thm_mainrevenuewelfare} are actually somewhat misleading as they represent the worst-case bound for each objective, which is derived independently from the other objective by simply bounding the worst-case revenue (or welfare) over all instances. However, the worst-case performance for revenue $(\zeta)$ and the worst-case performance for social welfare $(\frac{2 - \alpha}{1-\alpha})$ need not and as we show, \emph{do not} occur for the same instance. As a matter of fact, for a given instance, if the actual welfare obtained is close to the guarantee provided in Theorem~\ref{thm_mainrevenuewelfare}, the large gap between welfare and revenue as in Figure~\ref{figure:bicriteria_unitdemand} completely vanishes and the approximation factors coincide.

We first state the main structural claim that enables this trade-off.

\begin{claim}
\label{lem_bicritimprov2}
Suppose that a pricing algorithm $Alg$ satisfies $SW(Alg) \leq c_1 \pi(Alg)$,  $SW^* - SW(Alg) \leq c_2 \pi(Alg)$. Then for every instance there exists some $1 \leq c \leq c_2 + 1$ such that $Alg$ is a bicriteria approximation $(\min(cc_1, \frac{c c_2}{c-1}), c)$ for revenue and welfare respectively.
\end{claim}


Unfortunately, the designer has no control over the factor $c$ as its exact value depends on the particular instance. Applying this claim to Theorem~\ref{thm_mainrevenuewelfare} yields the following corollary.


\begin{corollary}
\label{clm_bicrittradeoff}
For every given instance, there exists a constant $1 \leq c \leq \frac{2-\alpha}{1-\alpha}$ so that the solution returned by the prices of Theorem~\ref{thm_mainrevenuewelfare} has a social welfare that is within a factor $c$ of the optimum welfare and such that
$$\pi^{opt}\leq SW^* \leq \min(\frac{c}{c-1}.\frac{1}{1-\alpha}, \zeta)\pi(Alg).$$
\end{corollary}


For example, for MHR demand functions the statement of Theorem \ref{thm_mainrevenuewelfare} makes it seem that this pricing scheme may return a solution which is a 2e-approximation for revenue and a 2-approximation for welfare. Corollary \ref{clm_bicrittradeoff} points out that the actual results are far {\em better}. As Figure \ref{figure:tradeoff} illustrates for MHR demand, tradeoffs between revenue and welfare actually guarantee that when revenue is far from optimum, welfare is very high, and vice versa.

\section{Consequences for Solutions with High Revenue}
\label{sec:implications}
In Sections \ref{sec:unitdemandresult} and \ref{sec:multi-minded}, we give efficient pricing mechanisms which simultaneously achieve good approximations for both revenue and welfare. Consider, however, a seller whose main priority is to simply maximize profits. This seller may choose to use a different pricing mechanism with better revenue guarantees than the ones offered in this paper. For example, the seller may choose prices which are guaranteed to come closer to achieving optimum revenue (these are efficiently computable for unit-demand settings~\cite{anshelevichKS15, hartlineK05} under certain additional assumptions), or even use a large amount of resources to solve the intractable problem of actually computing prices $p^{opt}$ which yield the highest possible revenue. One of the main messages of our paper is as follows:

\vskip 3pt
{\em No matter what pricing mechanism the seller uses to optimize revenue, they can instead use a pricing mechanism which guarantees at least $1/\zeta$ fraction of the optimum welfare, without sacrificing any revenue.}
\vskip 3pt

In this section, we use a simple albeit highly general framework to derive results of this form. Although this framework is defined in terms of our model, it is actually general and applies across a wide variety of markets. 

Recall that $SW^*$ denotes the optimum social welfare and $\pi^{opt}$ denotes the maximum achievable profit. Given a pricing algorithm $Alg$, we will refer to the social welfare of the solution returned by $Alg$ using $SW(Alg)$, and its profit by $\pi(Alg)$. Consider an arbitrary profit maximization algorithm $Alg$ that achieves a good approximation with respect to $\pi^{opt}$. How do we go about characterizing the social welfare at these solutions?  The following theorem uses an existing profit-maximization algorithm whose guarantees hold with respect to the optimum welfare as a black-box to bound the welfare due to $Alg$.



\begin{theorem}
\label{thm_crosscomparison}
Consider a benchmark profit maximization algorithm $Alg^b$ whose profit $\pi(Alg^b)$ is always within a $c$ factor of $SW^*$ for some $c \geq 1$. Consider any other pricing algorithm $Alg$ for the same class of valuations that obtains at least as much profit as guaranteed by $Alg^b$ on all instances. Then the social welfare obtained by $Alg$ is at least a factor $\frac{1}{c}$ times that of the optimum welfare.
\end{theorem}

The proof simply follows from the fact that $SW(Alg) \geq \pi(Alg) \geq \pi(Alg^b) \geq \frac{SW^*}{c}$.

\subsection*{Implications for Unit-Demand Markets}
We now apply the framework provided by Theorem~\ref{thm_crosscomparison} for unit-demand using our result from Theorem~\ref{thm_mainrevenuewelfare} as a black-box. 

\begin{claim}
\label{clm_crosscomparison_ud}
Let $Alg$ be any algorithm for unit-demand markets that obtains at least as much profit as guaranteed by the algorithm of Theorem~\ref{thm_mainrevenuewelfare} on all instances. Then the social welfare obtained by $Alg$ is at least a factor $\frac{1}{\zeta}$ times that of the optimum welfare.
\end{claim}

Thus, consider the case when a seller is using any arbitrary pricing mechanism $Alg$, with the main goal being to maximize profit. By simply computing the revenue given by our pricing schemes from Section \ref{sec:unitdemandresult}, and then choosing the one which guarantees better revenue (i.e., choosing between $Alg$ and our pricing scheme), we form a new pricing algorithm which does not sacrifice any revenue compared to $Alg$, and due to the above theorem, is {\em also} guaranteed to have good social welfare. Moreover, for sellers who are able to compute the prices which result in absolute maximum revenue, the following is a trivial consequence of the above theorem:

\begin{corollary}
The ratio of the optimum social welfare $SW^*$ to the social welfare at the maximum profit solution $SW(\vec{p^{opt}},\vec{x^{opt}},\vec{y^{opt}})$ is at most $\zeta = \Theta\left(\frac{1}{1-\alpha}\right).$
\end{corollary}

For example, for MHR demand functions ($\alpha=0$), this implies that even for sellers who only care about profits, there is essentially no excuse not to also guarantee at least $1/2e$ of the optimum social welfare. Thus, for the settings we consider, one can strive for truly high revenue, without sacrificing much in welfare.

%
%

\section{Multi-Minded Buyers}\label{sec:multi-minded}

We now move on to our most general case with multi-minded buyers, wherein every buyer wishes to purchase one bundle from a desired set (of bundles). We use $\ell^{max}$ to denote the cardinality of the maximum sized bundle desired by any buyer type, and $\ell^{min}$ for the minimum sized bundle. The main result in this section is a bicriteria approximation algorithm that extends our results for the unit-demand case. Our algorithm still achieves a $\Theta(\frac{1}{1-\alpha})$-approximation to the optimum welfare; as for profit, we obtain a $\Theta(\frac{1}{1-\alpha})$ bound further discounted by a $\log(\Delta)$ factor, where $\Delta = \frac{\ell^{max}}{\ell^{min}}$. 

Moreover, as in the previous section, our profit bound is obtained in terms of the optimum social welfare, which allows the consequences mentioned in Section \ref{sec:implications} as well as an analog of Corollary \ref{clm_bicrittradeoff} to hold showing that the true lower bounds are better than the theorem states.

\begin{theorem}\label{thm:multi}
For any given instance with multi-minded buyers, there exists a poly-time\\$\left(\Theta(\frac{\log (\Delta)}{1-\alpha}), \Theta(\frac{1}{1-\alpha})\right)$-bicriteria approximation algorithm for profit and welfare respectively.
\end{theorem}

\noindent{\bf Proof Sketch:} The proof for the general case with multi-minded buyers is quite involved (see Appendix \ref{app:multi}), so here we provide a high level overview of the various ingredients that combine to form the proof.

\vskip 2pt\noindent{\em Step 1: Benchmark Solution.} The first step involves defining a benchmark solution $(\vec{x^{b}}, \vec{y^{b}})$ whose social welfare, as in the unit-demand case, is at most a $\frac{2-\alpha}{1-\alpha}$ factor away from $SW^*$. Specifically, define $\tilde{p} = \lambda^{max} (1-\alpha)^{\frac{1}{\alpha}}$. Now, the benchmark demand vector $\vec{x^b}$ is defined as follows: for each buyer $i$, $x^b_i = \min \{x^*_i, \lambda^{-1}(\tilde{p})\}$. The allocation vector $\vec{y^b}$ is simply a scaled-down version of the optimum allocation vector, i.e., suppose that $x^*_i(S)$ is amount of bundle $S$ consumed by buyer type $i$ in $(\vec{p^*}, \vec{x^*}, \vec{y^*})$. Then, the amount of bundle $S$ consumed by $i$ in the benchmark solution $x^b_i(S) := \frac{x^b_i}{x^*_i}x^*_i(S)$. This consumption pattern gives rise to the allocation vector $\vec{y^b}$ on the goods. Let $SW^b$ denote the social welfare of this benchmark solution.

In the unit-demand case, we were able to identify a suitable price vector to actually implement this benchmark solution and extract a good fraction of its welfare as revenue. However, this is no longer be possible for the general case, as an analogous pricing solution would involve personalized payments for each buyer and thus may not be realizable using simple item pricing.  For convenience, we define $\pi^b = \sum_{i \in B}\lambda_i(x^b_i)x^b_i - C(\vec{y^b})$ to be the profit due to the personalized payments. Observe that if we allow for discriminatory prices, then the revenue maximization becomes almost trivial.



Despite its lack of implementability, the benchmark solution's appeal is due to the fact that both $SW^b$ and $\pi^b$ approximate $SW^*$ up to a $\Theta(\frac{1}{1-\alpha})$ factor, which we can be shown using similar arguments as in the unit-demand case. This is still a non-trivial claim; even with personalized payments, it is not clear if we can ever approximate the optimum social welfare via profit. Now that we have a benchmark solution which achieves good welfare and revenue, but is not implementable using item pricing, our goal in this proof becomes to compute item prices which approximate this benchmark solution in both revenue and welfare.

\vskip 2pt\noindent{\em Step 2: Augmented Walrasian Equilibrium.} Our goal will be to use $(\vec{x^b}, \vec{y^b})$ as a guide to design a sequence of (item) pricing solutions that together behave like the benchmark solution, and return the `best approximation' among these solutions. Towards this end, we introduce the notion of an \textit{Augmented Walrasian Equilibrium} with dummy price $p^d$, which is a social welfare maximizing solution for an augmented problem, consisting of the original instance plus $|T|$ dummy buyers having constant valuation $p^d$, one for each good $t \in T$. Our notion of an augmented Walrasian equilibrium can be viewed an extension of the Walrasian Equilibrium with reserve prices concept introduced in~\cite{guruswami2005profit} to settings where buyers purchase bundles of arbitrary sizes. Unlike the unit-demand case where there is a linear relationship between the reserve price and the revenue, bounding the profit in equilibrium solutions (once dummy buyers are removed) with multi-minded buyers and production costs require a careful characterization of the various goods. Specifically, there are really two types of goods in such solutions: ones that are {\em saturated}, i.e., $p_t=c_t(y_t)$, and ones which are artificially limited by the dummy price $p^d$, so $p_t=p^d$.

\vskip 2pt\noindent{\em Step 3: Starting dummy price.} The plan is to form a sequence of solutions resulting from Augmented Walrasian Equilibria with dummy prices, but what dummy price should we start with? Here we prove the following crucial claim. The main argument which makes this claim work is the ability to charge any deviation from the benchmark solution to the saturated goods in the current pricing solution, which are guaranteed to provide good welfare.

\begin{claim}
\label{clm_cruc_appbenchmark}
Suppose that $(\vec{p}, \vec{x}, \vec{y})$ denotes a pricing solution for the original instance obtained via an augmented Walrasian equilibrium with dummy price $\frac{\tilde{p}}{2 \ell^{max}}$. Then,

$$SW^* - SW(\vec{p}, \vec{x}, \vec{y}) \leq \left(5 + \frac{6}{1-\alpha}\right)[\pi(\vec{p}, \vec{x}, \vec{y}) + \pi(\vec{p^*}, \vec{x^*}, \vec{y^*})].$$
\end{claim}

\vskip 2pt\noindent{\em Step 4: Sequence of Solutions.} Previously, we identified a suitable dummy price and proved that the solution obtained via the augmented equilibrium at this dummy price (to some extent) captures most of the welfare of the benchmark solution. What about the profit of this solution? When all of the bundles desired by buyers are of equal cardinality, one can immediately show that the solution's profit also mimics that of the personalized payment scheme in the benchmark solution, $\pi^b$. When there is large disparity in the bundle sizes, however, the dummy price may result in buyers gravitating towards the smaller sized bundles. How do we fix this?	

We next consider augmented Walrasian equilibria at dummy prices that are scaled versions of the original dummy price $\frac{\tilde{p}}{2\ell^{max}}$, and once again, charge the lost welfare (due to a buyer facing high prices compared to her personalized payment in the benchmark solution) special class of saturated goods, which always yield high revenue. The profit due to the saturated goods is no larger than the social welfare of the solution, and thus, in either event, the social welfare cannot be small. More formally, we define a series of pricing solutions $(\vec{p}(j), \vec{x}(j), \vec{y}(j))$ for $j=1$ to $j = \delta = \lceil{\log(\Delta)}\rceil + 1$ to be the pricing solution for the original instance obtained via the augmented Walrasian equilibrium at dummy price $2^j \frac{\tilde{p}}{2\ell^{max}}$. Let us also define $SW(j)$ and $\pi(j)$ to be the social welfare and profit of the respective pricing solutions. Then, we have that:



\begin{claim}
\label{clm_diffinwelf}
For every $j \in [0, \delta]$, $SW(j) - SW(j+1) \leq 3\pi(j) + 3\pi(j+1)$.
\end{claim}

Since we can bound the difference in welfare via profit, we can sequentially build solutions with good social welfare until we either reach a solution with good profit or $j = 1 + \delta$ is reached, and we have no more solutions to construct. Now combining the above claims along with an upper bound on $SW(\delta+1)$ in terms of $\pi(\delta+1)$ (which we prove in the Appendix), we obtain a bound on $SW^*$ in terms of the profits of various solutions, which would imply that the max-profit solution approximates the optimum welfare up to the desired bound.

%

\begin{claim}
\label{clm_final_welfareintermsrev}
$$SW^* \leq \left(8 + 2(\frac{1}{1-\alpha})^{\frac{1}{\alpha}} + 4\frac{1}{1-\alpha}\right)[\sum_{j=0}^{1+\delta} \pi(j) + \pi(\vec{p^*}, \vec{x^*}, \vec{y^*})].$$
\end{claim}

\vskip 2pt\noindent{\em Step 5: Final Pricing Algorithm.} Unfortunately, the above argument is still not strong enough to give us a bicriteria bound where the welfare is independent of $\Delta$.
In the statement of Claim~\ref{clm_final_welfareintermsrev}, there are $2 + \log(\Delta)$ profit terms on the right hand side; therefore, the solution giving maximum profit among $\pi(0), \pi(1), \ldots, \pi(\delta+1)$ and $\pi(\vec{p^*}, \vec{x^*}, \vec{y^*})$ must yield a $(2+\log(\Delta))\Theta(\frac{1}{1-\alpha})$ approximation to both profit and welfare. 

In our final step of the proof, we take this proof further and show that if the max-profit solution from the above claim does not result in a good welfare, then one can in fact identify another such solution with similar profit guarantees but \emph{much better} welfare. Specifially, we show that by using some careful analysis, one can instead compute a solution whose approximation for social welfare is simply $\Theta(\frac{1}{1-\alpha})$ and thus, is independent of $\Delta$. We now define our main algorithm: for the purpose of continuity, let $\pi(-1)$ represent $\pi(\vec{p^*}, \vec{x^*}, \vec{y^*})$.

\begin{enumerate}
\item Let $k$ be the smallest index in the range $[-1, 1+\delta]$ such that $\frac{SW^*}{\pi(k)}$ is no larger than $2(\log(\Delta) + 2)\{\left(8 + 2(\frac{1}{1-\alpha})^{\frac{1}{\alpha}} + 4\frac{1}{1-\alpha}\right) \}$.

\item Return the solution $(\vec{p}(k), \vec{x}(k), \vec{y}(k))$.
\end{enumerate}

From Claim~\ref{clm_final_welfareintermsrev}, it is clear that there exists at least one index $k$ providing the desired guarantee. In fact, if $k=-1$, then we are done because the solution returned is the one maximizing social welfare. In our final claim (see Appendix), we show that for $k$ defined as above, we have that $SW^*$ is within a factor of $12(\frac{2-\alpha}{1-\alpha})$ of $SW(k)$. $\qed$- 

Applying structural Claim~\ref{lem_bicritimprov2} to the above proof, we get profit-welfare trade-offs analogous to Corollary~\ref{clm_bicrittradeoff}, namely that for every given instance, there exists a constant $1 \leq c \leq 12(\frac{2-\alpha}{1-\alpha})$, so that if $SW^* = c SW(Alg)$, then $SW^* \leq \frac{2c}{c-2} (5 + \frac{6}{1-\alpha}) \pi(Alg)$, when $c \geq 2$. For instance, this implies that either $SW(Alg) \geq \frac{SW^*}{3}$ or that $\pi(Alg)$ is actually a $\Theta(\frac{1}{1-\alpha})$-approximation to OPT. Here $Alg$ refers to our bicriteria algorithm from Theorem~\ref{thm:multi}.

\subsubsection*{Consequences for other high-revenue solutions.}
A direct application of Theorem~\ref{thm_crosscomparison} using our newly obtained bounds on multi-minded buyers as an intermediate yields the following claim.

\begin{claim}
Let $Alg$ be any algorithm that obtains at least as much profit as guaranteed by the algorithm of Theorem~\ref{thm:multi} on all instances. Then the social welfare obtained by $Alg$ is at most a factor $\Theta(\frac{\log(\Delta)}{1-\alpha})$ away from the optimum welfare.
\end{claim}

\section{Conclusion and Future Directions}
In this work, we were able to provide envy-free posted pricing algorithms that simultaneously approximate both profit and social welfare for markets with quite general buyer valuations and production costs. Such multi-objective algorithms are extremely well-motivated in a variety of realistic settings, where revenue and welfare are closely interconnected. We used our profit-maximization guarantees as a black-box and showed that any solution with reasonable profit guarantees (including the maximum profit solution) generates good welfare. In the process, we provide a partial characterization of the exact friction between these two objectives.



\bibliography{bibliography}
\bibliographystyle{plain}

\newpage
\appendix

\section{Appendix: Proof of Theorem~\ref{thm_mainrevenuewelfare} for Unit-Demand Buyers}
We begin by defining some notation pertinent to the proof. Given a price vector $\vec{p}$, we will use $q_i(\vec{p})$ to denote the price of the minimally priced good desired by buyer type $i$, i.e., $q_i(\vec{p}) = \min_{t \in B_i}p_t$. For the rest of the proof, we will use the terms buyer and buyer type interchangeably. Recall that under any solution $(\vec{p},\vec{x},\vec{y})$, if buyer $i$ purchases non-zero amount of good $t \in B_i$, then $p_t = q_i(\vec{p}) = \lambda_i(x_i)$. More specifically, we will only consider solutions that minimize the overall cost subject to the constraint that each buyer $i$ only purchases from the minimally priced goods in $S_i$. That is, given $\vec{p},\vec{x}$, there are several candidate allocation vectors $\vec{y}$ that correspond to valid distributions of the demand vector $\vec{x}$ on the goods. Unless mentioned otherwise, all of our solutions will consider allocation vectors that minimize the total cost $C(\vec{y}) = \sum_{t \in T}C_t(y_t)$ among the set of valid allocation vectors. Observe that in such a solution, if buyer $i$ purchases two different goods $t, t'$, then it has to be the case that the marginal cost of the two goods are equal, i.e., $c_t(y_t) = c_{t'}(y_{t'})$. Therefore, without loss of generality, we can use $r_i(\vec{y})$ to denote the marginal cost of any of the goods being used by buyer $i$ under the given allocation. We will also use $C(\vec{y})$ as a shortcut for the total cost incurred under a given allocation $\vec{y}$, i.e., $C(\vec{y}) = \sum_{t \in T}C_t(y_t)$.

\subsubsection*{The Social Welfare Maximizing Solution}
The optimum welfare solution $OPTW:= (\vec{p^*},\vec{x^*},\vec{y^*})$ is the pricing solution that maximizes the total social welfare of the system. The following simple claim (inspired by a similar claim from~\cite{anshelevichKS15}) provides a useful characterization of the optimum solution.

\begin{claim}
\label{clm_optchar}
There exists an optimum solution $(\vec{p^*},\vec{x^*},\vec{y^*})$ where for every good $t \in T$, $p^*_t = c_t(y^*_t)$. \end{claim}

We also define the notion of a welfare-maximizing solution constrained by demand. Specifically, an allocation vector $\vec{y}$ is said to be a welfare-maximizer with respect to given demand $\vec{x}$, if $\vec{y}$ is the minimum cost allocation consistent with the demand vector $\vec{x}$.

\subsubsection*{Computing the Welfare Maximizing Solution}
The algorithm described in Section~\ref{sec:unitdemandresult} explicitly uses the prices in the social welfare maximizing solution $OPTW$. In order to compute the prices efficiently, we use the following convex program to compute the welfare maximizing demand and allocation vector $(\vec{x^*}, \vec{y^*})$ and then apply Claim~\ref{clm_optchar} to obtain the corresponding price vector.

\begin{equation}
\centering
\begin{aligned}
\quad \text{max} & \sum_{i \in B}\int_{x=0}^{x_i}\lambda_i(x)dx - \sum_{t \in T}C_t(y_t)\\
\text{s.t.} & \sum_{t: t \in B_i}x_{it} = x_i \quad \forall i \in B\\
& \sum_{i: t \in B_i}x_{it} = y_t \quad \forall t \in T\\
& \quad \quad x_{it}, x_i, y_t \geq 0, \quad \forall i \in B, t \in T
\end{aligned}
\label{cp_welfaremaximization}
\end{equation}

We now proceed with the proof of the theorem. The proof conforms to the following structure: first, we define a special pricing strategy known as a thresholded pricing vector and prove that such a strategy satisfies some desirable properties. After this, we impose some sufficient conditions on such a pricing vector that allow us to prove bicriteria approximation bounds on the revenue and welfare. Finally, we prove that the thresholded pricing strategy presented in Section~\ref{sec:unitdemandresult} does indeed satisfy these sufficient conditions.


\begin{proof} The main instrument that allows us to design a simple but efficient pricing strategy is the notion of a \textit{thresholded pricing vector}. Such a pricing vector extends uniform pricing strategies (all goods have the same price), which have been successfully applied to design solutions with good revenue~\cite{balcanBM08} towards settings with production costs. We will show that thresholded pricing vectors enjoy many desirable properties and are convenient to analyze as opposed to the revenue maximizing solution, which is hard to get a grip on.

\begin{definition}{\textit{(Thresholded Pricing Vector)}}
A pricing vector $\vec{p}$ is said to be a thresholded pricing vector if $\exists$ a primary price $\tilde{p}$ such that for every good $t$, $p_t = \max(\tilde{p},p^*_t)$.
\end{definition}

We begin with a useful lemma that allows us to characterize the allocation resulting from a thresholded pricing vector.

\begin{lemma}[Lemma \ref{lem_proxysimilartoopt} from Section \ref{sec:unitdemandresult}]
Suppose that $((\tilde{p})_{t \in T}, (\tilde{x})_{i \in B}, (\tilde{y})_{t \in T})$ is a pricing solution resulting from a thresholded pricing vector. Then,
\begin{enumerate}
\item The market can be clustered into two mutually disjoint sets of buyers and goods $(B^H, T^H)$ and $(B^L, T^L)$ so that the buyers in each cluster only purchase the goods in the same cluster and $(a)$ for $(i,t) \in (B^H,T^H)$, $\tilde{p}_t = p^*_t$, $\tilde{x}_i = x^*_i$, and $\tilde{y}_t = y^*_t$; $(b)$ for $(i,t) \in (B^L,T^L)$, $\tilde{p}_t \geq p^*_t$, $\tilde{x}_i \leq x^*_i$, and $c_t(\tilde{y}_t) \leq c_t(y^*_t)$.

\item $\vec{\tilde{y}}$ is a welfare-maximizing allocation vector with respect to the demand vector $\vec{\tilde{x}}$.
\end{enumerate}
\end{lemma}
\begin{proof}
Divide the buyers and goods as follows: let $T^H$ be the set of goods with price (strictly) higher than $\tilde{p}$ and let $B^H$ be the buyers using these goods. Define $T^L$ as the goods with price $\tilde{p}$ and $B^L$ as the corresponding buyers. Since each buyer only purchases from the min-priced goods available to her, buyers in $B^L$ will only purchase (some of) the goods in $T^L$, and the same is true for those in $B^H$.

We begin by characterizing the solution corresponding to $(B^H,T^H)$: for all $t \in T^H$, $\tilde{p}_t = p^*_t$ as per the definition of $\vec{\tilde{p}}$. Consider any buyer $i \in B^H$; clearly no good in $T^L$ can belong to this buyer's desired set $B_i$ since she is using the minimally priced good(s) available to her. This price must be exactly $q_i(\vec{p^*})$ and thus her demand must also be $\tilde{x}_i = x^*_i$. Next, for every good $t' \in T^L$ and $t \in T^H$, $p^*_{t'} \leq \tilde{p} < p^*_t$. Therefore, every buyer who is using a good in $T^H$ in optimum welfare solution will still be using that good in our proxy solution. We can, therefore conclude that for all $t \in T^H$, $\tilde{y}_t = y^*_t$.

Finally, look at any good $t \in T^L$, which is priced at exactly $\tilde{p}$. By definition, $\tilde{p} \geq p^*_t$. Therefore, for every buyer $i \in T^L$, $q_i(\vec{p^*}) \leq \tilde{p} = q_i(\vec{\tilde{p}})$. And so, $\tilde{x}_i \leq x^*_i$. Since every buyer's demand is smaller in our solution w.r.t the max-welfare solution, we can use Lemma~\ref{lem_diffcostmonoton} just for the buyers and goods in $(B^L, T^L)$. We get that for every $t \in T^L$, $c_t(\tilde{y}_t) \leq c_t(y^*_t)$. \\

\noindent \textbf{(Part 2)}: In order to show that $\vec{\tilde{y}}$ has the smallest cost among all feasible allocations satisfying demand vector $\vec{\tilde{x}}$, it is sufficient (and necessary) to prove that for every buyer $i$, and every good $t$ such that $t \in B_i$, $r_i(\vec{\tilde{y}}) \leq c_t(\tilde{y}_t)$. Recall that every good $t$ is priced at either $p^*_t$ or at $\tilde{p}$ if $\tilde{p} > p^*_t$.

By definition, we have two min-cost sub-allocations: 1) buyers in $B^L$ are using the cheapest possible allocation using only the items in $T^L$ because all goods in $T^L$ have the same price and we consider only cost minimizing allocations for the given pricing solution; 2) the same is true for $B^H$ and $T^H$, this is because for these entities, the allocation is the same as in the max-welfare solution, which must definitely minimize cost for any given subset of the actual buyers and goods. So, in order to prove that this is a min-cost flow, we only need to consider the cross-edges, i.e., buyers belonging to $B^L$ and goods desired by these buyers belonging to $T^H$. We already know that for any given buyer in $B^H$ and good $t \in T^L$, $t \notin B_i$. So we can conveniently ignore this case. What about the reverse case, can there be a buyer $i$ in $B^L$ and an item $t$ in $T^H$ such that $c_t(\tilde{y}_t) = c_t(y^*_t) < r_i(\vec{\tilde{y}})$?

We know that for every $t' \in T^L$, $c_{t'}(\tilde{y}_{t'}) \leq c_{t'}(y^*_{t'})$. Since all buyers in $B^L$ are using these goods in both the current solution and in $(\vec{p^*},\vec{x^*},\vec{y^*})$, it must be that for every such buyer $i \in B^L$,  $r_i(\vec{\tilde{y}}) \leq r_i(\vec{y^*})$. But since $i$ does not use any good $t \in T^H$ even in the optimum solution, we get that $r_i(\vec{y^*}) \leq c_t(y^*_t) = c_t(\tilde{y}_t)$. This completes the proof.
\end{proof}

Now that we have a better understanding of thresholded pricing solutions, i.e., solutions resulting from thresholded pricing vectors, we are ready to prove our main theorem ala Lemma~\ref{lem_bicritimprov} by constructing a suitable threshold vector and establishing upper bounds on $SW(Alg)$ and $SW^* - SW(Alg)$. We divide the proof into three major claims: in the first two claims, we consider a general thresholded pricing solution and establish sufficient conditions that allow us to derive useful upper bounds on $\pi(Alg)$ as a function of the primary price $\tilde{p}$. Following this, we construct an actual thresholded pricing vector that satisfies the required sufficient conditions and complete the proof.

\begin{claim}
\label{clm_proxy}
Suppose that $(\vec{\tilde{p}}, \vec{\tilde{x}}, \vec{\tilde{y}})$ is a thresholded pricing solution with primary price $\tilde{p}$. 
Then, the total social welfare of this solution $SW(\vec{\tilde{p}}, \vec{\tilde{x}}, \vec{\tilde{y}})$ is at most a factor $2\frac{\lambda^{max}}{\tilde{p}} - 1$ times the profit due to this solution $\pi(\vec{\tilde{p}}, \vec{\tilde{x}}, \vec{\tilde{y}})$.
\end{claim}
\begin{proof}
The social welfare of the current solution is $\sum_{i \in B}u_i(\tilde{x}_i) - C(\vec{\tilde{y}})$. The function $u_i$ is concave for every $i \in B$, and therefore $u_i(\tilde{x}_i)\leq u'_i(0)\cdot\tilde{x}_i = \lambda^{max}\cdot\tilde{x}_i$. So, our first inequality is the following, 

%


$$\sum_{i \in B}u_i(\tilde{x}_i) - C(\vec{\tilde{y}}) \leq \sum_{i \in B}\lambda^{max} \tilde{x}_i - C(\vec{\tilde{y}}) = \frac{\lambda^{max}}{\tilde{p}}\sum_{i \in B} \tilde{p} \tilde{x}_i - C(\vec{\tilde{y}}) \leq \frac{\lambda^{max}}{\tilde{p}}\sum_{t \in T}\tilde{p}_t\tilde{y}_t - C(\vec{\tilde{y}}).$$

The final inequality comes from the fact that $\tilde{p} \leq \tilde{p}_t$ for all $t \in T$ and from rearranging the allocation from the buyers to the goods. Now, the total profit that the seller makes at the given strategy $\pi(\vec{\tilde{p}}, \vec{\tilde{x}}, \vec{\tilde{y}})$ equals $\sum_{t \in T}\tilde{p}_t \tilde{y}_t - C(\vec{y})$. Using this, we get the following upper bound for the ratio of the welfare to profit
\begin{align*}
\frac{\sum_{i \in B}u_i(\tilde{x}_i) - C(\vec{\tilde{y}})}{\pi(\vec{\tilde{p}}, \vec{\tilde{x}}, \vec{\tilde{y}})} & \leq \frac{\frac{\lambda^{max}}{\tilde{p}}\sum_{t \in T}\tilde{p}_t \tilde{y}_t - C(\vec{\tilde{y}})}{\sum_{t \in T}\tilde{p}_t \tilde{y}_t - C(\vec{y})}\\
& \leq \frac{\frac{\lambda^{max}}{\tilde{p}}\sum_{t \in T}\tilde{p}_t \tilde{y}_t - \sum_{t \in T}\frac{1}{2}c_t(\tilde{y}_t) \tilde{y}_t}{\sum_{t \in T}[\tilde{p}_t  - \frac{1}{2}c_t(\tilde{y}_t)]\tilde{y}_t}\\
& \leq \frac{\sum_{t \in T}\frac{\lambda^{max}}{\tilde{p}}\tilde{p}_t \tilde{y}_t - \frac{1}{2}\tilde{p}_t\tilde{y}_t}{\sum_{t \in T}[\tilde{p}_t \tilde{y}_t - \frac{1}{2}\tilde{p}_t \tilde{y}_t]} \\
& = 2\frac{\lambda^{max}}{\tilde{p}} - 1.
\end{align*}

The second inequality above comes from the definition of doubly convex cost functions according to which $C_t(\tilde{y}_t) \leq \frac{1}{2}c_t(\tilde{y}_t)$. The third inequality comes from the fact that for every $t$, $c_t(\tilde{y}_t) \leq \tilde{p}_t$. To see why this is true, observe from Lemma~\ref{lem_proxysimilartoopt} that for every $t \in T$,  $c_t(\tilde{y}_t) \leq c_t(y^*_t) = p^*_t \leq \tilde{p}_t$. This completes the proof of the first of our three claims.
\end{proof}

\begin{claim}
\label{clm_proxy2}
Suppose that $(\vec{\tilde{p}}, \vec{\tilde{x}}, \vec{\tilde{y}})$ is a thresholded pricing solution with primary price $\tilde{p}$ satisfying the following condition,
\begin{itemize}

\item For every $i \in B$, either $\frac{\lambda_i(\tilde{x}_i) - r_i(\vec{\tilde{y}})}{|\lambda'_i(\tilde{x}_i)|} \leq \tilde{x}_i$ or $\lambda_i(x^*_i) \geq \tilde{p}$.

\end{itemize}
Then, the difference in social welfare with respect to $OPTW$,i.e., $SW_2 := SW^* - SW(\vec{\tilde{p}}, \vec{\tilde{x}}, \vec{\tilde{y}})$ is at most a factor $\frac{1}{1-\alpha}$ times the profit due to this solution $\pi(\vec{\tilde{p}}, \vec{\tilde{x}}, \vec{\tilde{y}})$.
\end{claim}
\begin{proof}

We need to prove an upper bound on $SW^* - SW(\vec{\tilde{p}}, \vec{\tilde{x}}, \vec{\tilde{y}}) = \sum_{i \in B}(u_i(x^*_i) - u_i(\tilde{x}_i)) - (C(\vec{y^*}) - C(\vec{\tilde{y}}))$. Applying Lemma~\ref{lem_costdiff}, we get that $(C(\vec{y^*}) - C(\vec{\tilde{y}})) \geq \sum_{i \in B}r_i(\vec{\tilde{y}})(x^*_i - \tilde{x}_i)$. Writing the utility as an integral, the required welfare can be simplified as
$$SW_2 \leq \sum_{i \in B}[\int_{\tilde{x}_i}^{x^*_i}\lambda_i(x)dx - r_i(\vec{\tilde{y}})(x^*_i - \tilde{x}_i)] = \sum_{i \in B}\int_{\tilde{x}_i}^{x^*_i}(\lambda_i(x) - r_i(\vec{\tilde{y}}))dx.$$

Now, consider the function $f_i(x) = \lambda_i(x) - r_i(\vec{\tilde{y}})$. The second term is constant w.r.t $x$ and $\lambda_i$ is $\alpha$-SR and so from Lemma~\ref{lem_asrconst}, $f_i$ is also $\alpha$-SR for all $i \in B$. Next, we can use our crucial Lemma~\ref{lem_integral} to bound the integral as follows $\int_{\tilde{x}_i}^{x^*_i}f_i(x)dx \leq \frac{1}{1-\alpha}(\frac{f_i(\tilde{x}_i)}{|f'_i(\tilde{x}_i)|})(f_i(\tilde{x}_i) - f(x^*_i))$. Summing this up, one obtains,
\begin{align*}
SW_2 & \leq \frac{1}{1-\alpha} \sum_{i \in B}\frac{\lambda_i(\tilde{x}_i) - r_i(\vec{\tilde{y}})}{|\lambda'_i(\tilde{x}_i)|}(\lambda_i(\tilde{x}_i) - \lambda_i(x^*_i)).
\end{align*}

Let $B^L$ and $B^H$ be as in Lemma~\ref{lem_proxysimilartoopt}. For $i\in B^H$, we know that $\lambda_i(\tilde{x}_i) = \tilde{p}_i = p^*_i = \lambda_i(x^*_i)$, so the above expression equals 0. From the claim statement, we know that for $i\in B^L$, the quantity $\frac{\lambda_i(\tilde{x}_i) - r_i(\vec{\tilde{y}})}{|\lambda'_i(\tilde{x}_i)|}$ must be bounded from above by $\tilde{x}_i$. Thus, we have that,

\begin{align*}
SW_2 & \leq \frac{1}{1-\alpha} \sum_{i \in B^L} \tilde{x}_i(\lambda_i(\tilde{x}_i) - r_i(\vec{y^*})).
\end{align*}

Now, we can complete the proof using the fact that for all $i$, $\lambda_i(\tilde{x}_i)\tilde{x}_i$ is the total profit due to that buyer.

$$SW_2 \leq \frac{1}{1-\alpha} \sum_{t \in T} [\tilde{y}_t \tilde{p}_t - c_t(\tilde{y}_t)\tilde{y}_t] \leq \frac{1}{1-\alpha}[\sum_{t \in T} \tilde{y}_t \tilde{p}_t - C(\vec{\tilde{y}})] = \frac{1}{1-\alpha}\pi(\vec{\tilde{p}}, \vec{\tilde{x}}, \vec{\tilde{y}}).$$

%

\end{proof}

We have shown that if we can conjure up a nice thresholded price vector that satisfies certain properties, we can bound the social welfare of the optimum solution in terms of the profit obtained by the above solution. We now show that the price vector described in Section~\ref{sec:unitdemandresult} satisfies all of our requirements.


\begin{claim}
\label{clm_constructive}
We can compute in poly-time a solution $((\tilde{p})_{t \in T}, (\tilde{x})_{i \in B}, (\tilde{y})_{t \in T})$ satisfying the following two conditions,
\begin{enumerate}
\item $\exists$ a price $\tilde{p}$ such that for every $t \in T$,  $\tilde{p}_t = \max(p^*_t, \tilde{p})$.

\item For every $i \in B$, either $\frac{\lambda_i(\tilde{x}_i) - r_i(\vec{\tilde{y}})}{|\lambda'_i(\tilde{x}_i)|} \leq x_i$ or $\lambda_i(x^*_i) \geq \tilde{p}$.
\end{enumerate}
\end{claim}
\begin{proof}
Suppose that all the users have $\alpha$-strongly regular utilities. Then, set $\tilde{p} = (1-\alpha)^{\frac{1}{\alpha}}\lambda^{max}$. Construct the price vector $\vec{\tilde{p}}$ as described above. Moreover, let $\vec{\tilde{x}}$ and $\vec{\tilde{y}}$ be the corresponding buyer demand vector and (feasible) allocation vector at this price. All that is remaining is to prove that this solution satisfies the second condition mentioned in the claim.

Consider any $i \in B$: the min-priced good available to this buyer is priced at either $\tilde{p}$ or $q_i(\vec{p^*})$. In the $\tilde{p}$ case, we know that the user has a demand of $\tilde{x}_i$ units satisfying $\lambda_i(\tilde{x}_i) = \tilde{p} = \lambda^{max} (1-\alpha)^{\frac{1}{\alpha}}.$ By definition $\lambda^{max} = \lambda_i(0)$ for all $i$. Therefore, for every user facing a price of $\tilde{p}$ we have that
$$\lambda_i(0) = (\frac{1}{1-\alpha})^{1/\alpha}\lambda_i(\tilde{x}_i).$$

We can loosen the above inequality as follows: $\lambda_i(0) - r_i(\vec{\tilde{y}}) \geq (\frac{1}{1-\alpha})^{1/\alpha}(\lambda_i(\tilde{x}_i) - r_i(\vec{\tilde{y}})).$ Now take $f_i(x) = \lambda_i(x) - r_i(\vec{\tilde{y}})$ and apply corollary~\ref{corr_crucial} with respect to $\tilde{x}_i$. We get that $\frac{f_i(\tilde{x}_i)}{|f'_i(\tilde{x}_i)|} \leq \tilde{x}_i$, which is the required criterion.

For the second case, the proof follows trivially because the min-priced good $t$ available to the user $i$ has a price of $p^*_t > \tilde{p}$. This is the same as the min-priced good available to this user in the optimum solution and therefore $\lambda_i(\tilde{x}_i) = \lambda_i(x^*_i) = p^*_t \geq \tilde{p}$.

\end{proof}

\subsection*{Final Leg: Combining all of the Claims}
The rest of the proof follows from a direct application of our structural lemmas. Primarily, consider the pricing algorithm $Alg$, that chooses a pricing vector as defined in Claim~\ref{clm_constructive}. Then, from Claims~\ref{clm_proxy} and~\ref{clm_proxy2}, we get that $SW(Alg) \leq \left((2\frac{1}{1-\alpha})^{\frac{1}{\alpha}} -1 \right) \pi(Alg)$ and $SW^* - SW(Alg) \leq \frac{1}{1-\alpha}\pi(Alg)$. Applying Lemma~\ref{lem_bicritimprov} yields the required bicriteria result.

Finally, we claim that $\zeta = \Theta(\frac{1}{1-\alpha})$. To see why, consider the function $\frac{2(\frac{1}{1-\alpha})^{\frac{1}{\alpha}}}{\frac{1}{1-\alpha}}$. Upon differentiation, we infer that this function is non-increasing and its maximum value of $2e$ is obtained as $x \rightarrow 0$.

\end{proof}

\section{Appendix: Proof of Theorem \ref{thm:multi} for Multi-Minded Buyers}\label{app:multi}
\noindent{\bf Proof Sketch:} The proof for the general case with multi-minded buyers is quite involved and so we begin by providing a high level overview of the various ingredients that combine to form the proof. The first step involves defining a benchmark solution $(\vec{x^{b}}, \vec{y^{b}})$ whose social welfare, as in the unit-demand case, is at most a $\frac{2-\alpha}{1-\alpha}$ factor away from $SW^*$. In the unit-demand case, we were able to identify a suitable price vector to actually implement this benchmark solution and extract a good fraction of its welfare as revenue. However, this may no longer be possible for the general case as an analogous pricing solution would involve personalized payments for each buyer and thus may not be realizable using simple item pricing. In other words, the benchmark solution is only realizable by charging different prices to different buyers for the same good, instead of item pricing in which a good has a single price for all buyers. Charging such discriminatory prices would change the problem completely and give the seller much more power (for example, maximizing revenue becomes almost trivial, instead of NP-Complete). Our goal in this proof is therefore to compute item prices which approximate this benchmark solution in both revenue and welfare.

The rest of the proof does not really use the benchmark solution directly; instead our goal will be to use $(\vec{x^b}, \vec{y^b})$ as a guide to design a sequence of (item) pricing solutions that together behave like the benchmark solution, and return the `best approximation' among these solutions. Towards this end, we introduce the notion of an \textit{Augmented Walrasian Equilibrium} with dummy price $p$, which is a social welfare maximizing solution for an augmented problem, consisting of the original instance plus $|T|$ dummy buyers having constant valuation $p$, one for each good $t \in T$. One way to think about this is as a Walrasian equilibrium with a reserve price $p$ for each good, so that the prices are required to be above $p$. We show that these augmented equilibria have several desirable properties once the dummy buyers are removed from the final solution.

The starting point for our algorithm is the identification of a carefully chosen dummy price (see Claim~\ref{clm_cruc_appbenchmark})
so that the ensuing augmented equilibrium (minus the dummy buyers) captures the welfare of the benchmark solution to a large extent. Of course, the solution may not extract the required profit, and to correct this (Claim~\ref{clm_diffinwelf}), we show that simply scaling the dummy price by a factor of two yields a solution whose welfare still approximates the original solution. The repeated scaling gives us a sequence of $\log(\Delta)$ different pricing solutions. Finally, we show that at least one solution in this sequence has profit and welfare which both approximate that of the benchmark solution.

Both of the above claims make use of a charging argument that may be of independent interest, wherein the lost welfare (due to a buyer facing high prices compared to her personalized payment in the benchmark solution) is charged to a special class of goods, referred to as saturated goods, which always yield high revenue.

%
%

\begin{proof}
\subsubsection*{Benchmark Solution}
Recall that $(\vec{p^*}, \vec{x^*}, \vec{y^*})$ denotes the pricing solution achieving the optimum social welfare. Define $\tilde{p} = \lambda^{max} (1-\alpha)^{\frac{1}{\alpha}}$. Now, the benchmark demand vector $\vec{x^b}$ is defined as follows: for each buyer $i$, $x^b_i = \min \{ x^*_i, \lambda^{-1}(\tilde{p})\}$. The allocation vector $\vec{y^b}$ is simply a scaled-down version of the optimum allocation vector, i.e., suppose that $x^*_i(S)$ is amount of bundle $S$ consumed by buyer type $i$ in $(\vec{p^*}, \vec{x^*}, \vec{y^*})$. Then, the amount of bundle $S$ consumed by $i$ in the benchmark solution $x^b_i(S) := \frac{x^b_i}{x^*_i}x^*_i(S)$. This consumption pattern gives rise to the allocation vector $\vec{y^b}$ on the goods. Let $SW^b$ denote the social welfare of this benchmark solution.

Does there exist a pricing solution $\vec{p^b}$ with a single price per good that implements the above benchmark solution? Unfortunately, one can easily define instances where this is not the case. That said, a trivial way to implement the benchmark solution is via a personalized payment scheme, i.e., each buyer $i$ is charged a price of $\lambda_i(x^b_i)$ that is completely independent of other buyers. For convenience, we define $\pi^b = \sum_{i \in B}\lambda_i(x^b_i)x^b_i - C(\vec{y^b})$ to be the profit due to the personalized payments. We begin by highlighting an obvious property of the benchmark solution that will come in handy later on.

\begin{proposition}
\label{prop_bench}
The benchmark solution vectors are dominated by the optimum solution vectors, i.e., $\vec{x^*} \geq \vec{x^b}$ and $\vec{y^*} \geq \vec{y^b}$. Therefore for every good $t$, $c_t(y^b_t) \leq c_t(y^*_t)$.
\end{proposition}

Despite its lack of implementability, the benchmark solution's appeal lies in the fact that both $SW^b$ and $\pi^b$ approximate $SW^*$ up to a $\Theta(\frac{1}{1-\alpha})$ factor, which we formally state below. Observe that this is still a non-trivial claim; even with personalized payments, it is not clear if we can ever approximate the optimum social welfare via profit. For instance, for the max-welfare solution $(\vec{x^*}, \vec{y^*})$, a personalized charging scheme achieves the same profit as $\pi(\vec{p^*}, \vec{x^*}, \vec{y^*})$, which may be quite poor. Our goal however, will be to use $(\vec{x^b}, \vec{y^b})$ as a guide to design a sequence of (item) pricing solutions that together behave like the benchmark solution, and return the `best approximation' among these solutions.


\begin{claim}
(1) $SW^b \leq \left[2(\frac{1}{1-\alpha})^\frac{1}{\alpha} -1 \right] \pi^b$~~~~~~~~~~~
(2) $SW^* - SW^b \leq \frac{1}{1-\alpha}\pi^b$.
\end{claim}

The proof of the above claim is extremely similar to the proofs of Claims~\ref{clm_proxy} and~\ref{clm_proxy2}, and we do not explicitly prove it here as we do not really use the claim anywhere in the rest of this proof. Instead, we will state and prove approximate versions of the above claim that we will require later. These approximate versions shed additional light on some of the sufficient conditions that a pricing solution must satisfy in order to capture social welfare via profit.

\begin{claim}
\label{clm_aboveandbelowthresh}
\begin{enumerate}
\item Suppose that $(\vec{p}, \vec{x}, \vec{y})$ is a pricing solution satisfying $\lambda_i(x_i) \geq \tilde{p}$ for every $i$, and $p_t \geq c_t(y_t)$ for all $t \in T$. Then,
$$SW(\vec{p}, \vec{x}, \vec{y}) \leq [2(\frac{1}{1-\alpha})^{\frac{1}{\alpha}}-1]\pi(\vec{p}, \vec{x}, \vec{y}).$$

\item Suppose that $\vec{x}$ is a demand vector satisfying $x_i \geq x^b_i$ for every $i$. Then,
$$\sum_{i \in B}\int_{x_i}^{x^*_i}\lambda_i(x)dx \leq \frac{1}{1-\alpha}\sum_{i \in B}\lambda_i(x_i)x_i.$$

\end{enumerate}
\end{claim}
\begin{proof}
\textbf{(Statement 1):}\\
The social welfare of $(\vec{p}, \vec{x}, \vec{y})$ is $\sum_{i \in B}u_i(x_i) - C(\vec{y})$. The function $u_i$ is concave for every $i \in B$, and therefore $u_i(x_i)\leq u'_i(0)\cdot x_i = \lambda^{max}x_i$. So, our first inequality is the following, 

%


$$\sum_{i \in B}u_i(x_i) - C(\vec{y}) \leq \sum_{i \in B}\lambda^{max} x_i - C(\vec{y}) = \frac{\lambda^{max}}{\tilde{p}}\sum_{i \in B} \tilde{p} x_i - C(\vec{y}) = (\frac{1}{1-\alpha})^{\frac{1}{\alpha}}\sum_{i \in B} \tilde{p} x_i - C(\vec{y}).$$

 Now, the total profit that the seller makes at the given solution $\pi(\vec{p}, \vec{x}, \vec{y})$ equals $\sum_{i \in B}\lambda_i(x_i)x_i - C(\vec{y}) \geq \sum_{i \in B}\tilde{p}x_i - C(\vec{y})$. Using this, we get the following upper bound for the ratio of the welfare to profit
\begin{align*}
\frac{\sum_{i \in B}u_i(x_i) - C(\vec{y})}{\pi(\vec{p}, \vec{x}, \vec{y})} & \leq \frac{(\frac{1}{1-\alpha})^{\frac{1}{\alpha}}\sum_{i \in B}\lambda_i(x_i) x_i - C(\vec{y})}{\sum_{i \in B}\lambda_i(x_i)x_i - C(\vec{y})}\\
& \leq \frac{(\frac{1}{1-\alpha})^{\frac{1}{\alpha}}\sum_{i \in B}\lambda_i(x_i) x_i - \frac{1}{2}\sum_{i \in B}\lambda_i(x_i)x_i}{\sum_{i \in B}\lambda_i(x_i) x_i - \frac{1}{2}\sum_{i \in B}\lambda_i(x_i)x_i}\\
& = 2(\frac{1}{1-\alpha})^{\frac{1}{\alpha}} - 1.
\end{align*}

The second inequality above stems from Lemma~\ref{lem_app_costintermlambda} in the Appendix, which states that for any feasible pricing solution, $C(\vec{y}) \leq \frac{1}{2}\sum_{i \in B}\lambda_i(x_i)x_i$. This completes the proof of the first statement.

\textbf{Statement 2:}

We need to prove an upper bound on the quantity $\sum_{i \in B}\int_{x_i}^{x^*_i}\lambda_i(x)dx$ given the condition that $x_i \geq x^b_i$ for all $i \in B$. Since $\lambda_i$ is an $\alpha$-strongly regular function for all $i$, we can apply our crucial Lemma~\ref{lem_integral} to bound the integral as follows $\int_{x_i}^{x^*_i}\lambda_i(x)dx \leq \frac{1}{1-\alpha}(\frac{\lambda_i(x_i)}{|\lambda'_i(x_i)|})(\lambda_i(x_i) - \lambda(x^*_i))$. Summming this up, one obtains,
\begin{align*}
\sum_{i \in B}\int_{x_i}^{x^*_i}\lambda_i(x)dx & \leq \frac{1}{1-\alpha} \sum_{i \in B'}\frac{\lambda_i(x_i)}{|\lambda'_i(x_i)|}(\lambda_i(x_i) - \lambda_i(x^*_i)) \\
& \leq \frac{1}{1-\alpha} \sum_{i \in B'} x_i\lambda_i(x_i).
\end{align*}

In the above expression, $B'$ is the set of buyers for whom $\lambda_i(x_i) > \lambda_i(x^*_i)$. For the remaining buyers, the integral is not positive and so we can safely ignore these terms. Further, since $x_i \geq x^b_i$, and $\lambda_i(x_i) > \lambda_i(x^*_i)$ for all $i\in B'$, it must be that $\lambda_i(x^b_i)>\lambda_i(x^*_i)$, and so by definition of the benchmark solution, we have that $\lambda(x^b_i)=\tilde{p}$. Thus, it is not hard to see that for every buyer $i \in B'$, $\lambda_i(x_i) \leq \lambda_i(x_i^b) = \tilde{p} = \lambda^{max}(1-\alpha)^{\frac{1}{\alpha}}$, and so, as per Corollary~\ref{corr_crucial}, $\lambda_i(x_i)/|\lambda'_i(x_i)| \leq x_i$. This completes the proof.
\end{proof}

\subsubsection*{Augmented Walrasian Equilibrium with Dummy Prices}
The most important notion in this proof is that of an augmented Walrasian equilibrium. At a high level, the augmented equilibrium is a useful tool that allows us to approximate `personalized payment'-based solutions via item pricing as long as the total price charged to different buyers is correlated. For instance, in the benchmark solution, buyers are either charged $\tilde{p}$ or $\lambda_i(x^*_i)$.

\begin{definition}
Given an instance $\mathcal{I} = (B,T)$, an augmented Walrasian equilibrium with dummy price $p$ is the social welfare maximizing allocation for an instance $\mathcal{I^+} = (B \cup B', T)$ where $B,T$ are as defined previously and $B'$ consists of $|T|$ buyer types, one for each $t \in T$, such that the buyer $i_t$ corresponding to good $t$ has only good $t$ in her desired set of bundles (exactly one bundle consisting of one item). Moreover, this buyer has a constant valuation $\lambda_{i_t}(x) = p$ for $x \leq L$, where $L$ is some sufficiently large population, and $\lambda_{i_t}(x) = 0$ for $x > L$.
\end{definition}

One can imagine that each new buyer type has a virtually infinite population of buyers all of whom desire one good and have the same valuation. Given $\mathcal{I}, p$, an augmented Walrasian equilibrium with this dummy price can be computed efficiently using the same convex program as before.

\noindent \textbf{Pricing Solutions Obtained via Augmented Equilibria} Given a dummy price $p$, the augmented equilibrium computation gives us a social welfare maximizing demand and allocation vector $(\vec{x^p}, \vec{y^p})$ for the corresponding instance $\mathcal{I^+}$. However, we are only interested in valid pricing solutions for our original instance $\mathcal{I}$. We now provide a simple transformation to convert $(\vec{x^p}, \vec{y^p})$ into a pricing solution $(\vec{p}, \vec{x}, \vec{y})$ for the original instance, which we will take for granted in the rest of this work. Essentially, such a pricing solution is obtained by taking the augmented Walrasian equilibrium, pricing each good at its marginal cost, and then removing all consumption due
to dummy buyers. Formally, define $p_t = c_t(y_t^p)$ for all $t \in T$, and for $i \in B$, $x_i = x_i^p$. Finally, let these buyers consume the same bundles as in the augmented equilibrium, i.e., for every $S \in B_i$, the amount of this bundle consumed by the buyer in $(\vec{p}, \vec{x}, \vec{y})$, $x_i(S)$ is the same as the amount consumed by this buyer in the augmented equilibrium, $x_i^p(S)$.

It is evident from the definition that $(\vec{p}, \vec{x}, \vec{y})$ is a valid pricing solution for our original instance $\mathcal{I}$. We first highlight some simple properties obeyed by such a pricing solution irrespective of the exact dummy price and in the process, gain some understanding of the structure of these solutions.

\begin{lemma}
\label{lem_propaugmented}
Given an instance and a dummy price $p^d$, let $(\vec{p}, \vec{x}, \vec{y})$ denote the pricing solution for the original instance obtained via the augmented Walrasian equilibrium with dummy price $p^d$. Then,

\begin{enumerate}
\item For every good $t \in T$, we have $p_t=\max\{p^d,c_t(y_t)\}$.

\item The profit due to the solution $\pi(\vec{p}, \vec{x}, \vec{y}) \geq C(\vec{y})$.
\end{enumerate}
\end{lemma}
\begin{proof}
Suppose that $(\vec{x^{p^d}}, \vec{y^{p^d}})$ denotes the augmented Walrasian equilibrium with dummy price $p^d$. By construction, we know that for every $t$, $p_t = c_t(y^{p^d}_t)$.
\begin{enumerate}
\item We first claim that for every good $t$, $c_t(y^{p^d}_t) \geq p^d$. Indeed, if this condition were not met for some $t$, then one could increase the social welfare of $(\vec{x^{p^d}}, \vec{y^{p^d}})$ by simply allocating a non-zero amount of good $t$ to the dummy buyer assigned to that good $i_t$ violating the fact that $(\vec{x^{p^d}}, \vec{y^{p^d}})$ is a social welfare maximizing allocation. Therefore for every $t$, it is true that $p_t = c_t(y^{p^d}_t) \geq p^d$.

Since buyers are using the same bundles as in the augmented equilibrium, the consumption of any good in $\vec{y}$ cannot be larger than its consumption in $\vec{y^{p^d}}$, i.e., for all $t \in T$, $y_t \leq y^{p^d}_t$. This implies that $p_t = c_t(y^{p^d}_t) \geq c_t(y_t)$.

Now consider some good $t$ for which $p_t > p^d$. We claim that no dummy buyer can be consuming this good in the augmented equilibrium. Once again, we show this via contradiction. Assume to the contrary: then, one could increase the social welfare of $(\vec{x^{p^d}}, \vec{y^{p^d}})$ by removing a non-zero amount of the allocation given to dummy buyer $i_t$, which violates the optimality of the solution. This means that the entire consumption of good $t$ comes from the original buyers in $B$, and therefore, $y_t = y^{p^d}_t$. This in turn implies that $c_t(y_t) = p_t$.

\item We have already showed that for all $t$, $p_t \geq c_t(y_t)$. Therefore, the second statement is a consequence of Lemma~\ref{lem_app_costintermlambda}.
\end{enumerate}

\end{proof}

\noindent\textbf{Saturated Goods} Perhaps, the most useful feature of pricing solutions obtained via augmented equilibria is highlighted in the second statement of the above lemma. Any good that is not priced at the dummy price is in a sense \textit{saturated} as $p_t = c_t(y_t)$. We refer to such goods having $p^d < p_t = c_t(y_t)$ as saturated goods. Saturated goods usually yield good profit (due to high consumption) and thus, are useful in bounding the social welfare in terms of profit. As we will show later, for a suitably chosen dummy price, any welfare lost from the benchmark solution can be charged to the income obtained only from these saturated goods.

\subsubsection*{First shot at an actual dummy price}
Now that we have a better understanding of pricing solutions obtained via augmented equilibria, we pose our first major question: can we effectively approximate the benchmark solution using a carefully chosen dummy price? Our journey towards answering the above question begins with the answer to a slightly relaxed question: what is the highest possible dummy price so that resulting solution has a social welfare comparable to that of the benchmark solution? Clearly, if the dummy price is small enough, one can recover the original optimum solution. The insistence on high prices is to ensure that each buyer's demand in the augmented equilibria is somewhat comparable to $\lambda_i(x^b_i)$. The following three claims shed some light on this question, by identifying $\frac{\tilde{p}}{2\ell^{max}}$ as the appropriate dummy price. At this dummy price, we show that any deviation from the benchmark solution can be charged to the revenue provided by the saturated goods (and so the deviation cannot be too large). Recall that $\ell^{max}$ is the size of the largest bundle desired by any buyer type.

\begin{lemma}
\label{lem_deviationsaturated}
Suppose that $(\vec{x^a}, \vec{y^a})$ is any given benchmark solution, satisfying $\lambda_i(x^a_i) \geq \hat{p}>0$ for all $i \in B$. Let  $(\vec{p}, \vec{x}, \vec{y})$ denote a pricing solution for the original instance obtained via an augmented Walrasian equilibrium with dummy price $\frac{\hat{p}}{2 \ell^{max}}$. Then, for any buyer $i$ having $\lambda_i(x_i) > \lambda_i(x^a_i)$, and every bundle $S$ desired by $i$ ($S \in B_i$), we have that,

\begin{itemize}
\item $S$ contains at least one saturated good.

\item $\lambda_i(x_i) \leq 2\sum_{t \in S \cap \hat{S}}p_t$, where $\hat{S}$ is the set of saturated goods in $(\vec{p}, \vec{x}, \vec{y})$.

\end{itemize}
\end{lemma}

\begin{proof}
The first part of the lemma follows from simple contradiction. Consider a buyer $i$ and bundle $S \in B_i$ as mentioned in the statement of the lemma, and suppose that $S$ does not contain any saturated good. Therefore, as per Lemma~\ref{lem_propaugmented}, it must be true that for every $t \in S$, $p_t = \frac{\hat{p}}{2\ell^{max}}$. But since we have a valid pricing solution, it is also true that $\lambda_i(x_i) \leq \sum_{t \in S}p_t$. So, we have that

$$ \hat{p} \leq \lambda_i(x^a_i) < \lambda_i(x_i) \leq \sum_{t \in S}p_t = |S| \frac{\hat{p}}{2\ell^{max}} \leq \frac{\hat{p}}{2},$$

which is a contradiction. Note that $|S| \leq \ell^{max}$, by definition.

Now that we have established that $S$ contains at least one saturated good, we can write $\lambda_i(x_i)$ as,
$$\lambda_i(x_i) \leq \sum_{t \in S \cap \hat{S}}p_t + \sum_{t \in S \setminus (S \cap \hat{S})}p_t \leq \sum_{t \in S \cap \hat{S}}p_t + \frac{\hat{p}}{2}.$$

Since $\lambda_i(x_i) \geq \hat{p}$, we get the desired lemma (second statement).
\end{proof}


Building on the previous lemma, we actually bound the welfare loss from the benchmark solution in terms of the income (profit without the cost) due to the saturated goods in $(\vec{p}, \vec{x}, \vec{y})$ as well as $(\vec{p^*}, \vec{x^*}, \vec{y^*})$. The introduction of $OPTW$ here may appear strange at first but we remind the reader that bounding social welfare in terms of the profit at $OPTW$ does not really hurt us by much as a high profit at $OPTW$ leads to a trivial bicriteria approximation algorithm with optimal welfare.

\begin{lemma}
\label{lem_boundviasaturated}
Suppose that $(\vec{x^a}, \vec{y^a})$ is any given benchmark solution and $(\vec{p}, \vec{x}, \vec{y})$ denotes a pricing solution for the original instance obtained via an augmented Walrasian equilibrium with dummy price $\frac{\hat{p}}{2 \ell^{max}}$. Also, suppose that the following conditions are satisfied

\begin{enumerate}
\item $\lambda_i(x^a_i) \geq \hat{p}$ for all $i \in B$.

\item $\vec{x^*} \geq \vec{x^a}$ and $\vec{y^*} \geq \vec{y^a}$.

\end{enumerate}

Let $\hat{B}$ be the set of buyers for whom $\lambda_i(x_i) > \lambda_i(x^a_i)$. If $\hat{S}$ denotes the set of saturated goods in $(\vec{p}, \vec{x}, \vec{y})$, then,

$$\sum_{i \in \hat{B}} \int_{x_i}^{x^a_i} \lambda_i(x)dx \leq \sum_{i \in \hat{B}}\lambda_i(x_i)x^a_i \leq  2\sum_{t \in \hat{S}}p_t y_t + 2\sum_{t \in \hat{S}}p^*_t y^*_t $$

\end{lemma}
\begin{proof}

Suppose that $x^a_i(S)$ is the amount of bundle $S$ consumed by buyer $i$ in the benchmark solution. Then, we have that

$$\sum_{i \in \hat{B}}\lambda_i(x_i)x^a_i = \sum_{i \in \hat{B}}\sum_{S \in B_i} \lambda_i(x_i) x^a_i(S).$$

Moreover, we know that for every $S \in B_i$, $\lambda_i(x_i) \leq 2\sum_{t \in S \cap \hat{S}}p_t$, from Lemma~\ref{lem_deviationsaturated}. Using this upper bound, we get

$$\sum_{i \in \hat{B}}\lambda_i(x_i)x^a_i \leq \sum_{i \in \hat{B}}\sum_{S \in B_i}\sum_{t \in S \cap \hat{S}} 2p_t x^a_i(S) \leq 2\sum_{t \in \hat{S}}p_t y^a_t.$$

Now, we divide the goods in $\hat{S}$ into two categories: set $\hat{S}^a$ represents the goods for which $y_t \geq y^a_t$, and the rest of the goods fall under $\hat{S}^*$. Moreover, for every good $t \in \hat{S}^*$, $c_t(y^*_t) \geq c_t(y^a_t) \geq c_t(y_t) = p_t$. Using this, we can proceed with our final simplification.

$$2\sum_{t \in \hat{S}}p_t y^a_t \leq 2\sum_{t \in \hat{S}^a}p_ty^a_t + 2\sum_{t \in \hat{S}^*}p_ty^a_t \leq 2\sum_{t \in \hat{S}^a}p_ty_t + 2\sum_{t \in \hat{S}^*}p^*_ty^*_t.$$

Since $\hat{S}^*, \hat{S}^a \subseteq \hat{S}$, we have our lemma.
\end{proof}

With our properties of augmented Walrasian equilibrium solutions firmly in place, we are now ready to prove our first crucial claim, a lower bound for the social welfare of the solution $(\vec{p}, \vec{x}, \vec{y})$ obtained via an augmented equilibrium in terms of the optimum welfare as well as the profit at $(\vec{p}, \vec{x}, \vec{y})$ as well as at $OPTW$. Assuming for the time being that the profit at $OPTW$ is not that large, the claim below gives us a handle on a (high enough) dummy price for which our solution behaves like the benchmark solution. This dummy price will be the starting point, from which we design a sequence of solutions to extract good profit. If the current dummy price was much smaller (say), then we would still be guaranteed a good social welfare but unfortunately, good profit becomes an arduous task.

\begin{claim}[Claim \ref{clm_cruc_appbenchmark} from Section \ref{sec:multi-minded}]
Suppose that $(\vec{p}, \vec{x}, \vec{y})$ denotes a pricing solution for the original instance obtained via an augmented Walrasian equilibrium with dummy price $\frac{\tilde{p}}{2 \ell^{max}}$. Then,

$$SW^* - SW(\vec{p}, \vec{x}, \vec{y}) \leq \left(5 + \frac{6}{1-\alpha}\right)[\pi(\vec{p}, \vec{x}, \vec{y}) + \pi(\vec{p^*}, \vec{x^*}, \vec{y^*})].$$
\end{claim}
Recall that $\tilde{p} = \lambda^{max}(1-\alpha)^{\frac{1}{\alpha}}$.
\begin{proof}
Consider $SW^* - SW(\vec{p}, \vec{x}, \vec{y})$; this can be expressed as follows,

$$SW^* - SW(\vec{p}, \vec{x}, \vec{y}) = \sum_{i \in B}\int_{x_i}^{x^*_i} \lambda_i(x)dx - (C(\vec{y^*}) - C(\vec{y})).$$

Recall our original benchmark solution $(\vec{x^b}, \vec{y^b})$. Divide $B$ into two subsets: $\hat{B}$ denotes the set of buyers for whom $\lambda_i(x_i) > \lambda_i(x^b_i)$ and $B \setminus \hat{B}$ denotes the other buyers. We now further simplify the above equation,

\begin{align*}
\sum_{i \in B}\int_{x_i}^{x^*_i} \lambda_i(x)dx & = \sum_{i \in \hat{B}}\int_{x_i}^{x^*_i} \lambda_i(x)dx + \sum_{i \in B \setminus \hat{B}}\int_{x_i}^{x^*_i} \lambda_i(x)dx\\
& = \sum_{i \in \hat{B}}\int_{x_i}^{x^b_i} \lambda_i(x)dx + [\sum_{i \in \hat{B}}\int_{x^b_i}^{x^*_i} \lambda_i(x)dx + \sum_{i \in B \setminus \hat{B}}\int_{x_i}^{x^*_i} \lambda_i(x)dx]\\
& \leq \sum_{i \in \hat{B}}\int_{x_i}^{x^b_i} \lambda_i(x)dx + \frac{1}{1-\alpha}[\sum_{i \in \hat{B}}x^b_i \lambda_i(x^b_i) + \sum_{i \in B \setminus \hat{B}}x_i \lambda_i(x_i)] \\
& \leq \left(1 + \frac{1}{1-\alpha}\right) \sum_{i \in \hat{B}}\lambda_i(x_i)x^b_i + \frac{1}{1-\alpha}\sum_{i \in B \setminus \hat{B}}x_i \lambda_i(x_i).
\end{align*}

The bound for the term in the square parentheses in lines two and three comes from the second half of Lemma~\ref{clm_aboveandbelowthresh} for which we define a new demand vector $\vec{x'}$ such that $x'_i = \max(x^b_i, x_i)$. The last line follows from simple rearrangement. The benchmark solution satisfies all of the conditions mentioned in Lemma~\ref{lem_boundviasaturated} and therefore, the first term above $\sum_{i \in \hat{B}}\lambda_i(x_i)x^b_i$ is no larger than $2\sum_{t \in \hat{S}}p_t y_t + 2\sum_{t \in \hat{S}}p^*_t y^*_t$. Once again, $\hat{S}$ is the set of saturated goods in our pricing solution.

In summary, we get

$$\sum_{i \in B}\int_{x_i}^{x^*_i} \lambda_i(x)dx \leq \left(2 + \frac{3}{1-\alpha}\right)\sum_{i \in B}[\lambda_i(x_i)x_i + \lambda_i(x^*_i)x^*_i].$$

Coming back to our original expression, the difference in social welfare can be written as

\begin{align*}
SW^* - SW(\vec{p}, \vec{x}, \vec{y}) & \leq \left(2 + \frac{3}{1-\alpha}\right)\sum_{i \in B}[\lambda_i(x_i)x_i + \lambda_i(x^*_i)x^*_i] + C(\vec{y}) \\
& \leq \left(4 + \frac{6}{1-\alpha}\right)[\pi(\vec{p}, \vec{x}, \vec{y}) + \pi(\vec{p^*}, \vec{x^*}, \vec{y^*})] + \pi(\vec{p}, \vec{x}, \vec{y}). \\
\end{align*}

The last step 
comes from Lemma~\ref{lem_app_costintermlambda}.
\end{proof}

\subsubsection*{Penultimate Leg: Sequence of Solutions}
Previously, we identified a suitable dummy price and proved that the solution obtained via the augmented equilibrium at this dummy price (to some extent) captures all of the welfare of the benchmark solution and any lost welfare can be charged to the saturated goods. What about the profit of this solution? When all of the bundles desired by buyers are of equal cardinality, one can immediately show that the solution's profit also mimics that of the personalized payment scheme in the benchmark solution, $\pi^b$. When there is large disparity in the bundle sizes, however, the dummy price may result in buyers gravitating towards the smaller sized bundles. How do we fix this?	

In this subsection, we consider augmented Walrasian equilibria at dummy prices that are scaled versions of the original dummy price $\frac{\tilde{p}}{2\ell^{max}}$. Sequentially, we show that, if the profit at the previous solution (starting with the original dummy price) does not meet our requirement, then scaling the dummy price by a factor of two, does not really lead to a large loss in social welfare. Once again, our primary technique is a charging argument that assigns the lost welfare to the profit of the saturated goods. The profit due to the saturated goods is no larger than the social welfare of the solution, and thus, in either event, the social welfare cannot be small.

We begin with some notation. Consider the vector of prices defined as $P(j) = 2^j \frac{\tilde{p}}{2\ell^{max}}$ for $j=0$ up to $j=\lceil{\log(\Delta)}\rceil + 1$, and define $(\vec{p}(j), \vec{x}(j), \vec{y}(j))$ to be the pricing solution for the original instance obtained via the augmented Walrasian equilibrium at dummy price $P(j)$. While below we assume that $\Delta$ is a power of 2, this argument is easily generalized. Let us also define $SW(j)$ and $\pi(j)$ to be the social welfare and profit of the pricing solution $(\vec{p}(j), \vec{x}(j), \vec{y}(j))$. Observe that $(\vec{p}(0), \vec{x}(0), \vec{y}(0))$ denotes the solution with the original dummy price for whose social welfare we obtained a bound in Claim~\ref{clm_cruc_appbenchmark}.

Our next claim establishes a bound on the loss in social welfare as we increase the dummy price, in terms of the the previous and the current profit. For the rest of this proof, we define $\delta := \log_2 (\Delta)$.

\begin{claim}[Claim \ref{clm_diffinwelf} from Section \ref{sec:multi-minded}]
For every $j \in [0, \delta]$, $SW(j) - SW(j+1) \leq 3\pi(j) + 3\pi(j+1)$.
\end{claim}
\begin{proof}
Consider $SW(j) - SW(j+1)$, this can be bound as follows,

$$SW(j) - SW(j+1) \leq \sum_{i \in B}\lambda_i(x_i(j+1))x_i(j) - (C(\vec{y}(j)) - C(\vec{y}(j+1)).$$

Let us divide $B$ into two groups of buyers: $B_1$, which are the buyers satisfying $\lambda_i(x_i(j+1)) \leq 2\lambda_i(x_i(j))$ and $B_2$ which represent the buyers for whom $\lambda_i$ has more than doubled. Our proof will proceed by charging the loss in welfare for the buyers in $B_1$ to $\pi(j)$; for the buyers in $B_2$, the loss in welfare will be charged to both $\pi(j)$ as well as the income/revenue from the saturated goods in $\pi(j+1)$. In fact, we are only interested in a special class of saturated goods, for which $c_t(y_t(j+1)) \geq c_t(y_t(j))$. Formally, define $SS(j+1)$ as the subset of saturated goods in $(\vec{p}(j+1), \vec{x}(j+1), \vec{y}(j+1))$ having $c_t(y_t(j+1)) \geq c_t(y_t(j))$.

We first state a simple sub-claim, which we prove later.

\begin{lemma}
\label{sublem_charging}
For every buyer $i \in B_2$, and any bundle $S \in B_i$ that $i$ consumes a non-zero amount of in $(\vec{p}(j), \vec{x}(j), \vec{y}(j))$, we have that $\lambda_i(x_i(j+1)) - 2\lambda_i(x_i(j)) \leq \sum_{t \in SS(j+1) \cap S}p_t(j+1)$.
\end{lemma}

Now, we are ready to proceed with the proof. Let us divide $\sum_{i \in B}\lambda_i(x_i(j+1))x_i(j)$ into two parts, namely $(i)$: $\sum_{i \in B_1}\lambda_i(x_i(j+1))x_i(j) + \sum_{i \in B_2}2\lambda_i(x_i(j))x_i(j)$, and $(ii)$: $\sum_{i \in B_2}[\lambda_i(x_i(j+1)) - 2\lambda_i(x_i(j))]x_i(j)$. Consider the first part,

$$\sum_{i \in B_1}\lambda_i(x_i(j+1))x_i(j) + \sum_{i \in B_2}2\lambda_i(x_i(j))x_i(j) \leq 2\sum_{i \in B}\lambda_i(x_i(j))x_i(j).$$

We proceed with the second part carefully. Consider the solution $(\vec{p}(j), \vec{x}(j), \vec{y}(j))$, and define $z_i(S)$ to be the amount of bundle $S$ consumed by buyer $i$ in this solution. Without loss of generality, assume that $B_i$ consists only of bundles that $i$ is consuming non-zero amounts of in $(\vec{p}(j), \vec{x}(j), \vec{y}(j))$. Applying our Sub-Lemma~\ref{sublem_charging}, we get that,

\begin{align*}
\sum_{i \in B_2}[\lambda_i(x_i(j+1)) - 2\lambda_i(x_i(j))]x_i(j) & \leq \sum_{i \in B_2}\sum_{S \in B_i}[\lambda_i(x_i(j+1)) - 2\lambda_i(x_i(j))] z_i(S).\\
& \leq \sum_{i \in B_2}\sum_{S \in B_i} \sum_{t \in SS(j+1) \cap S}p_t(j+1)z_i(S).\\
& \leq \sum_{t \in SS(j+1)}p_t(j+1)y_t(j)\\
& \leq \sum_{t \in SS(j+1)}p_t(j+1)y_t(j+1).
\end{align*}

Finishing up the proof, we have that

\begin{align*}
SW(j) - SW(j+1) & \leq 2\sum_{i \in B}\lambda_i(x_i(j))x_i(j) - C(\vec{y}(j)) + \sum_{t \in T}p_t(j+1)y_t(j+1) + C(\vec{y}(j+1))\\
& = \pi(j) +\sum_{i \in B}\lambda_i(x_i(j))x_i(j) + \sum_{t \in T}p_t(j+1)y_t(j+1) + C(\vec{y}(j+1))\\
& \leq 3\pi(j) + 3\pi(j+1).
\end{align*}

As usual, we replace the income with the profit as per Lemma~\ref{lem_app_costintermlambda}. Recall that $\sum_{t \in T}p_t(j+1)y_t(j+1) = \sum_{i \in B}\lambda_i(x_i(j+1))x_i(j+1)$.

\subsubsection*{Proof of Lemma~\ref{sublem_charging}}
We need to prove that for every buyer $i \in B_2$, and any bundle $S \in B_i$ that $i$ consumes a non-zero amount of in $(\vec{p}(j), \vec{x}(j), \vec{y}(j))$,
$$\lambda_i(x_i(j+1)) - 2\lambda_i(x_i(j)) \leq \sum_{t \in SS(j+1) \cap S}p_t(j+1).$$

Consider the items in $S \setminus SS(j+1)$. Since $\lambda_i(x_i(j+1))=\sum_{t \in S}p_t(j+1)$ because $i$ consumes a non-zero amount of bundle $S$, then in order to prove the lemma, it suffices to show that

$$\sum_{i \in S \setminus SS(j+1)}p_t(j+1) \leq 2\sum_{t \in S \setminus SS(j+1)}p_t(j) \leq 2\sum_{t \in S}p_t(j) = 2\lambda_i(x_i(j)).$$

We claim that for every $t \in S \setminus SS(j+1)$, $p_t(j+1) \leq 2p_t(j)$. To see why this is true, assume by contradiction that $p_t$ has more than doubled for some good $t \in S \setminus SS(j+1)$. Clearly, this good must be saturated in  $(\vec{p}(j+1), \vec{x}(j+1), \vec{y}(j+1))$ since the dummy price has only doubled, and $c_t(y_t(j+1) \geq 2c_t(y_t(j))$. However, saturated goods which have $c_t(y_t(j+1) \geq c_t(y_t(j))$ belong to $SS(j+1)$ by definition.
\end{proof}

Since we can bound the difference in welfare via profit, we can sequentially build solutions with good social welfare until we either reach a solution with good profit (in which case, we are done) or $j = 1 + \delta$ is reached, and we have no more solutions to construct. In order to analyze the boundary condition, we now establish upper bounds on the social welfare at $j=1+\delta$ in terms of the profit obtained at that solution.

\begin{lemma}\label{lem_welfarelastbound}
$$SW(\delta+1) \leq \left(2(\frac{1}{1-\alpha})^{\frac{1}{\alpha}} - 1 \right)\pi(\delta+1).$$
\end{lemma}
\begin{proof}
Consider the pricing solution $(\vec{p}(\delta+1), \vec{x}(\delta+1), \vec{y}(\delta+1))$: this corresponds to an augmented Walrasian equilibrium with dummy price $2^{\delta+1} \frac{\tilde{p}}{2\ell^{max}} = \frac{\tilde{p}}{\ell^{min}}$, where $\ell^{min} \geq 1$ is the cardinality of the smallest bundle desired by any one buyer type.

Applying Lemma~\ref{lem_propaugmented}, we know that for every $t \in T$, $p_t \geq \frac{\tilde{p}}{\ell^{min}}$. Therefore, we have that for every buyer $i \in B$, $\lambda_i(x_i(\delta+1)) \geq \frac{\tilde{p}}{\ell^{min}} |S|$, for any $S \in B_i$. Since $|S| \geq \ell^{min}$, we conclude that $\lambda_i(x_i(\delta+1)) \geq \tilde{p}$. Applying Claim~\ref{clm_aboveandbelowthresh} on $(\vec{p}(\delta+1), \vec{x}(\delta+1), \vec{y}(\delta+1))$ completes the proof.
\end{proof}

Now, we are ready to combine all of the above claims and lemmas to make our final stand: a bound on $SW^*$ in terms of the profits of various solutions, which would imply that the max-profit solution approximates the optimum welfare up to the desired bound. In the next subsection, we take this proof further and show that if the max-profit solution from the below claim does not result in a good welfare, then one can in fact identify another such solution with similar profit guarantees but \emph{much better} welfare.

\begin{claim}[Claim \ref{clm_final_welfareintermsrev} from Section \ref{sec:multi-minded}]

$$SW^* \leq \left(8 + 2(\frac{1}{1-\alpha})^{\frac{1}{\alpha}} + 4\frac{1}{1-\alpha}\right)[\sum_{j=0}^{1+\delta} \pi(j) + \pi(\vec{p^*}, \vec{x^*}, \vec{y^*})].$$

\end{claim}
\begin{proof}

We know from Claim~\ref{clm_cruc_appbenchmark} that
$$SW^* - SW(0) \leq \left(5 + \frac{6}{1-\alpha}\right)[\pi(0) + \pi(\vec{p^*}, \vec{x^*}, \vec{y^*})].$$

Our other crucial result, Claim~\ref{clm_diffinwelf} tells us that for any $j$,

$$SW(j) - SW(j+1) \leq 3\pi(j) + 3\pi(j+1).$$

Writing the above inequality for all $j \in [0,\delta]$ provides us with a nice telescoping summation, which upon addition with the first inequality above yields

$$SW^* \leq SW(\delta+1) + \sum_{j=1}^{\delta}6\pi(j) + 3\pi(\delta+1) + \left(8 + \frac{6}{1-\alpha}\right)[\pi(0) + \pi(\vec{p^*}, \vec{x^*}, \vec{y^*})].$$

Finally, by replacing $SW(\delta+1)$ with the bound obtained in Lemma~\ref{lem_welfarelastbound}, we can wrap up the claim,

\begin{align*}
SW^* & \leq (2 + 2(\frac{1}{1-\alpha})^{\frac{1}{\alpha}})\pi(\delta+1) + 6\sum_{j=1}^{\delta}\pi(j) + \left(8 + \frac{6}{1-\alpha}\right)[\pi(0) + \pi(\vec{p^*}, \vec{x^*}, \vec{y^*})] \\
\leq & \left(8 + 2(\frac{1}{1-\alpha})^{\frac{1}{\alpha}} +  \frac{4}{1-\alpha}\right)[\sum_{j=0}^{\delta+1}\pi(j) + \pi(\vec{p^*}, \vec{x^*}, \vec{y^*})]
\end{align*}
For the second inequality, we used the fact that $2\frac{1}{1-\alpha} \leq 2(\frac{1}{1-\alpha})^{\frac{1}{\alpha}}$.
\end{proof}

\subsection*{Final Pricing Algorithm}
We have all of the required lemmas in our arsenal, and now will go about deriving our final algorithm. In the statement of Claim~\ref{clm_final_welfareintermsrev}, there are $2 + \log(\Delta)$ profit terms in the RHS; therefore, the solution giving maximum profit among $\pi(0), \pi(1), \ldots, \pi(\delta+1)$ and $\pi(\vec{p^*}, \vec{x^*}, \vec{y^*})$ must yield a $(2+\log(\Delta))\Theta(\frac{1}{1-\alpha})$ approximation to both profit and welfare. Our proof does not end here as our desired social welfare bound does not depend on $\Delta$, but by choosing the maximum profit solution, we may end up with a $O(\log(\Delta))$ additional welfare degradation on our hands.

However, we show that by using some careful analysis, one can instead compute a solution whose approximation for social welfare is simply $O(\frac{1}{1-\alpha})$ and thus, is independent of $\Delta$. We now define our main algorithm: for the purpose of continuity, let $\pi(-1)$ represent $\pi(\vec{p^*}, \vec{x^*}, \vec{y^*})$.

\begin{enumerate}
\item Let $k$ be the smallest index in the range $[-1, 1+\delta]$ such that $\frac{SW^*}{\pi(k)}$ is no larger than $2(\log(\Delta) + 2)\{\left(8 + 2(\frac{1}{1-\alpha})^{\frac{1}{\alpha}} + 4\frac{1}{1-\alpha}\right) \}$.

\item Return the solution $(\vec{p}(k), \vec{x}(k), \vec{y}(k))$.
\end{enumerate}

From Claim~\ref{clm_final_welfareintermsrev}, it is clear that there exists at least one index $k$ providing the desired guarantee. In fact, if $k=-1$, then we are done because the solution returned is the one maximizing social welfare. The following claim provides a lower bound on the social welfare $SW(k)$ when $k \geq 0$.

\begin{claim}
The algorithm provides a $(\Theta(\frac{\log \Delta}{1-\alpha}), \Theta(\frac{1}{1-\alpha}))$-bicriteria approximation to revenue and welfare.
\end{claim}
\begin{proof}
The revenue bound simply follows from the algorithm, so we focus on the social welfare here. Recall from Claim~\ref{clm_cruc_appbenchmark} that $(i) SW^* \leq SW(0) + \left(5 + \frac{6}{1-\alpha}\right)(\pi(-1) + \pi(0))$.

Moreover, applying Lemma~\ref{clm_diffinwelf} repeatedly for $j=0$ to $k$ and performing the telescoping summation, we get that $(ii): SW(0) - SW(k) \leq 3\pi(0) + 6\sum_{j=1}^{k-1}\pi(j) + 3\pi(k)$.

We use the above inequalities to derive a lower bound on $SW(k)$. We prove this in two cases,

\textbf{Case I:} $k = 0$

\begin{align*}
SW^* & \leq SW(0) + \left(5 + \frac{6}{1-\alpha}\right)(\pi(-1) + \pi(0))\\
\implies \frac{SW^*}{2} & \leq 6\left(1 + \frac{1}{1-\alpha}\right) SW(0).
\end{align*}

Clearly $\pi(-1) (5+\frac{6}{1-\alpha}) \leq \pi(-1) (8+\frac{4}{1-\alpha} + 2(\frac{1}{1-\alpha})^{\frac{1}{\alpha}}) \leq \frac{SW^*}{2}$ as per the definition of the algorithm since $k \neq -1$. Moreover, we also used the trivial inequality that $\pi(0) \leq SW(0)$. In conclusion, we have that $SW(k) = SW(0)$ is a $\Theta(\frac{1}{1-\alpha})$-approximation to the optimum welfare.

\textbf{Case II:} $k > 0$.

We add up the inequalities $(i)$ and $(ii)$ from above to get

\begin{align*}
SW^* & \leq SW(k) + 3\pi(k) + (8 + \frac{6}{1-\alpha})\sum_{j=-1}^{k-1}\pi(j)\\
& \leq SW(k) + 3\pi(k) + \frac{SW^*}{2}\\
& \leq 4SW(k) + \frac{SW^*}{2} \\
\implies \frac{SW^*}{2} & \leq 4SW(k).
\end{align*}

The second inequality is a consequence of the fact that for every $j \in [-1, k-1]$, $SW^* \geq 2(\log(\Delta) + 2)\{\left(8 + 2(\frac{1}{1-\alpha})^{\frac{1}{\alpha}} + 4\frac{1}{1-\alpha}\right)\pi(j)$. We add this inequality for $j=-1$ to $k-1$, which consists of at most $k+1 \leq \log(\Delta) + 2$ terms.

So, the worst welfare is obtained for $k =0$.

\end{proof}

\end{proof}

\section{Appendix: Cost functions}
The first two lemmas below are paraphrased versions of similar results in~\cite{anshelevichKS15}.
\begin{lemma}
\label{lem_diffcostmonoton}
Consider two pricing solutions $(\vec{p_1}, \vec{x_1}, \vec{y_1})$ and $(\vec{p_2}, \vec{x_2}, \vec{y_2})$ such that $\vec{p_1} \geq \vec{p_2}$, and suppose that $\vec{y_1}$, $\vec{y_2}$ are the minimum-cost allocations for the buyer demands $\vec{x_1}$ and $\vec{x_2}$ respectively over all feasible allocations. Then,
\begin{enumerate}
\item $\vec{x_2} \geq \vec{x_1}$, i.e, every buyer type's demand is higher under $\vec{p_2}$ than under $\vec{p_1}$.

\item For all $t \in S$, $c_t(y^1_t) \leq c_t(y^2_t)$.

\item For every buyer $i$, $r_i(\vec{z^1}) \leq r_i(\vec{z^2})$.
\end{enumerate}
\end{lemma}

\begin{lemma}
\label{lem_costdiff}
Consider two buyer demands $\vec{x^1}$ and $\vec{x^2}$ where $\vec{x^2} \geq \vec{x^1}$ and the corresponding min-cost allocations are $\vec{z^1}$ and $\vec{z^2}$. Then, the following provides a bound for the difference in costs
$$\sum_{t \in T} (C_t(z^2_t) - C_t(z^1_t)) \geq \sum_{i \in B} r_i(\vec{z^1}) (x^2_i - x^1_i).$$
\end{lemma}

\begin{lemma}
\label{lem_app_costintermlambda}
Let $(\vec{p}, \vec{x}, \vec{y})$ denote a valid pricing solution for a given instance such that for all $t$, $p_t \geq c_t(y_t)$. Then
\begin{enumerate}
\item $\sum_{i \in B}\lambda_i(x_i)x_i \geq 2C(\vec{y})$.

\item $\pi(\vec{p}, \vec{x}, \vec{y}) \geq C(\vec{y})$.

\item $\sum_{i \in B}\lambda_i(x_i)x_i \leq 2\pi(\vec{p}, \vec{x}, \vec{y})$.
\end{enumerate}

\end{lemma}
\begin{proof}
The proof follows from the nature of doubly convex cost functions. Double convexity implies that for every $t \in T$, $C_t(y_t) \leq \frac{1}{2}c_t(y_t)y_t$. Therefore, we have that
$$C(\vec{y}) \leq \frac{1}{2}\sum_{t \in T}c_t(y_t)y_t \leq \frac{1}{2}\sum_{t \in T}p_ty_t = \frac{1}{2}\sum_{i \in B}\lambda_i(x_i)x_i.$$
The final part of the above inequality comes from rearranging the allocation from the goods to the buyers, along with an application of the fact that for any buyer type $i$ with non-zero consumption of bundle $S$, $\lambda_i(x_i) = \sum_{t \in S}p_t$.

The second statement follows directly from the first since $\pi(\vec{p}, \vec{x}, \vec{y}) = \sum_{i \in B}\lambda_i(x_i)x_i - C(\vec{y})$. The third statement follows from the second.
\end{proof}

\section{Appendix: Basic Properties of $\alpha$-SR functions}

\begin{lemma}
\label{lem_asrconst}
Let $f(x)$ be any non-increasing, non-negative $\alpha$-Strongly regular function. Then for any constant $c$, $f(x)-c$ is also $\alpha$-strongly regular in the interval where it is non-negative.
\end{lemma}
\begin{proof}
We need to prove that for any $x_1 < x_2$, $\frac{f(x_2)-c}{|f'(x_2)|} - \frac{f(x_1)-c}{|f'(x_1)|} \leq \alpha(x_2 - x_1)$. We show this in two cases.

Case I: $|f'(x_1)| \geq |f'(x_2)|$\\
When both of the derivatives are non-zero, we have that
$$\frac{f(x_2)-c}{|f'(x_2)|} - \frac{f(x_1)-c}{|f'(x_1)|} = \{\frac{f(x_2)}{|f'(x_2)|} - \frac{f(x_1)}{|f'(x_1)|}\} - c(\frac{1}{|f'(x_2)|} - \frac{1}{|f'(x_2)|}).$$
Since the function $f$ is $\alpha$-SR, the first quantity in the RHS is smaller than or equal to $\alpha(x_2 - x_1)$ the second quantity is not positive as per the assumption for Case I. The lemma follows. Note that if any one of the two derivatives are zero, then the proof follows trivially.

Case II: $|f'(x_1)| < |f'(x_2)|$\\
In this case,

$$\frac{f(x_2)-c}{|f'(x_2)|} \leq \frac{f(x_2)-c}{|f'(x_1)|} \leq \frac{f(x_1)-c}{|f'(x_1)|}.$$
The final inequality follows from the fact that the function $f$ is monotone non-increasing. Therefore, $\frac{f(x_2)-c}{|f'(x_2)|} - \frac{f(x_1)-c}{|f'(x_1)|} \leq 0 \leq \alpha(x_2 - x_1)$.

\end{proof}

\begin{lemma}
\label{lem_integral}
Let $f(x)$ be any non-increasing, non-negative $\alpha$-Strongly regular function. Then for any $x_2 \geq x_1$, the following inequality is true,
$$\int_{x_1}^{x_2}f(x)dx \leq \frac{1}{1-\alpha}\frac{f(x_1)}{|f'(x_1)|}(f(x_1) - f(x_2)).$$
\end{lemma}
\begin{proof}
Consider any $x \geq x_1$. We now use the definition of $\alpha$-strong regularity to obtain a crude upper bound on $f(x)$ in terms of $f(x_1)$ and the derivative at $x$. i.e., we know that
$$\frac{f(x)}{|f'(x)|} \leq \frac{f(x_1)}{|f'(x_1)|} + \alpha(x - x_1).$$

For convenience, define $h = \frac{f(x_1)}{|f'(x_1)|}$. This gives us the following upper bound, $f(x) \leq h |f'(x)| + \alpha|f'(x)|(x - x_1).$ The quantity we wish to compute is $\int_{x_1}^{x_2}f(x)dx$. Using our upper bound from before, we have that
$$\int_{x=x_1}^{x_2}f(x)dx \leq h\int_{x=x_1}^{x_2}|f'(x)|dx + \alpha \int_{x=x_1}^{x_2}(x - x_1)|f'(x)|dx.$$

Recall that since $f$ is non-increasing $|f'(x)| = -f'(x)$. Moreover, for the second integral, we can simply integrate by parts and get that $\int_{x_1}^{x_2}(x - x_1)(-f'(x))dx = [(x - x_1)(-f(x))]_{x_1}^{x_2} - \int_{x=x_1}^{x_2}(-f(x))dx = -(x_2 - x_1)f(x_2) + \int_{x=x_1}^{x_2}f(x)$. Combining the various pieces, we are now ready to complete the proof
\begin{align*}
\int_{x_1}^{x_2}f(x)dx & \leq h(f(x_1) - f(x_2)) - \alpha(x_2 - x_1)f(x_2) + \alpha\int_{x=x_1}^{x_2}f(x)dx \\
\implies (1-\alpha)\int_{x_1}^{x_2}f(x)dx & \leq h(f(x_1) - f(x_2)) - \alpha(x_2 - x_1)f(x_2)\\
& \leq h(f(x_1) - f(x_2)).
\end{align*} \end{proof}

\begin{lemma}
\label{lem_hazardthresh}
Let $f(x)$ be a non-increasing, non-negative $\alpha$-Strongly regular function and let $x_1$ be any point satisfying $\frac{f(x_1)}{|f'(x_1)|} \leq x_1$. Then for every $x \geq x_1$, the following inequality holds,
$$\frac{f(x)}{|f'(x)|} \leq x.$$
\end{lemma}
\begin{proof}
The proof follows almost directly from the definition of an $\alpha$-SR function. Since $x \geq x_1$, we have that
$$\frac{f(x)}{|f'(x)|} \leq \frac{f(x_1)}{|f'(x_1)|} + \alpha(x-x_1) \leq x_1 + \alpha(x-x_1) \leq (1-\alpha)x_1 + \alpha x.$$

Since $x_1 \leq x$, we have that $(1-\alpha)x_1 + \alpha x \leq x$. This completes the proof.
\end{proof}

\begin{lemma}
\label{lem_subclaim_mhr}
Let $f(x)$ be a non-increasing, non-negative $\alpha$-Strongly regular function, and let $x_1$ be any point satisfying $\frac{f(x_1)}{|f'(x_1)|} > x_1$. Then,
$$f(0) < f(x_1)(\frac{1}{1-\alpha})^{1/\alpha}.$$

\end{lemma}
Notice that since $f$ is non-increasing, $f(0)$ is the maximum value of the function.
\begin{proof}
Consider the function $g(x) = \log(f(x))$. Differentiating this gives us,
$$g'(x) = \frac{f'(x)}{f(x)}.$$
Recall that since $f(x)$ is non-increasing, its derivative cannot be positive. Integrating $g'(x)$ from $x_1$ to $0$ gives us
\begin{align*}
\int_{x_1}^0 g'(x) = & \int_{x_1}^0 \frac{f'(x)}{f(x)}dx\\
g(0) - g(x_1) = & -\int_{x_1}^0 \frac{|f'(x)|}{f(x)}dx\\
\log\frac{f(0)}{f(x_1)} = &  \int_{0}^{x_1} \frac{|f'(x)|}{f(x)}dx\\
\end{align*}

Now, since the function is $\alpha$-SR and $x \leq x_1$, we have that $\frac{f(x)}{|f'(x)|} \geq \frac{f(x_1)}{|f'(x_1)|} - \alpha(x_1 - x)$. The first quantity in RHS is at least $x_1$ from the lemma statement, and therefore, we get $\frac{f(x)}{|f'(x)|} > x_1(1-\alpha) + \alpha x$. Substituting this in the above integral, we get

$$\int_{0}^{x_1} \frac{|f'(x)|}{f(x)}dx < \int_{0}^{x_1}\frac{1}{x_1(1-\alpha) + \alpha x}dx = \frac{1}{\alpha}\log(\frac{1}{1-\alpha}).$$

Therefore, we get that $\log\frac{f(0)}{f(x_1)} < \frac{1}{\alpha}\log(\frac{1}{1-\alpha})$, giving us the upper bound of $(\frac{1}{1-\alpha})^{\frac{1}{\alpha}}$ for $\frac{f(0)}{f(x_1)}$.
\end{proof}

The following corollary that follows from the previous lemma is perhaps the most important ingredient required for our main theorem.

\begin{corollary}
\label{corr_crucial}
Let $f(x)$ be a non-increasing, non-negative $\alpha$-Strongly regular function and $x_1$ be some point satisfying $f(0) \geq (\frac{1}{1-\alpha})^{\frac{1}{\alpha}}f(x_1)$. Then, $\frac{f(x_1)}{|f'(x_1)|} \leq x_1$.
\end{corollary}
\begin{proof}
Assume by contradiction that $\frac{f(x_1)}{|f'(x_1)|} > x_1$. Then as per Lemma~\ref{lem_subclaim_mhr}, it must be true that $f(0) < f(x_1)(\frac{1}{1-\alpha})^{\frac{1}{\alpha}}$, which of course contradicts the corollary statement.
\end{proof}

\end{document}